\documentclass[11pt]{article}
\usepackage{amsmath}
\usepackage{amsthm}
\usepackage{amssymb}
\usepackage{mathrsfs}
\usepackage[linesnumbered,ruled,vlined]{algorithm2e}
\usepackage{subfig}
\usepackage{color}
\usepackage[english]{babel}
\usepackage{graphicx}
\usepackage{wrapfig,epsfig}
\usepackage{epstopdf}
\usepackage{url}
\usepackage{graphicx}
\usepackage{color}
\usepackage{epstopdf}
\usepackage{algpseudocode}
\usepackage{scrextend}
\usepackage[T1]{fontenc}
\usepackage{bbm}
\usepackage{comment}
\usepackage{multicol}{}
\usepackage{wasysym}
\DeclareMathAlphabet{\mathpzc}{OT1}{pzc}{m}{it}

\usepackage{tikz}
\usepackage{hyperref}
\hypersetup{colorlinks=true,citecolor=red,linkcolor=blue}
\usetikzlibrary{arrows}
\usepackage[margin=1in]{geometry}
\graphicspath{{./figs/}}
\usepackage{wrapfig}
\usepackage{makecell}
\usepackage{setspace}

\usepackage{thmtools}
\usepackage{thm-restate}

\usepackage{sectsty}
\allsectionsfont{\boldmath}

\usepackage{multirow}

\newtheorem{lemma}{Lemma}[section]

\newenvironment{proof-sketch}{\noindent{\textit{Proof Sketch.}}\hspace*{1em}}{\qed\bigskip}

\usepackage[capitalize]{cleveref}
\crefname{theorem}{Theorem}{Theorems}
\Crefname{lemma}{Lemma}{Lemmas}
\Crefname{figure}{Figure}{Figures}

\newcommand{\nm}{n_{\max}}

\newcommand{\wt}{\widetilde}

\renewcommand{\varepsilon}{\epsilon}
\renewcommand{\tilde}{\wt}

\newcommand{\ID}{\mathsf{ID}}
\newcommand{\ceil}[1]{\left\lceil #1 \right\rceil}
\newcommand{\strongcd}{\mathsf{Strong}\text{-}\mathsf{CD}}
\newcommand{\sendercd}{\mathsf{Sender}\text{-}\mathsf{CD}}
\newcommand{\receivercd}{\mathsf{Receiver}\text{-}\mathsf{CD}}
\newcommand{\nocd}{\mathsf{No}\text{-}\mathsf{CD}}

\newcommand{\idle}{\mathsf{idle}}
\newcommand{\listen}{\mathsf{listen}}
\newcommand{\transmit}{\mathsf{transmit}}

\newcommand{\randomized}{\mathsf{rand}}
\newcommand{\deterministic}{\mathsf{det}}

\DeclareMathOperator*{\E}{{\mathbb{E}}}

\DeclareMathOperator{\poly}{poly}

\definecolor{b2}{RGB}{51,153,255}
\definecolor{mygreen}{RGB}{80,180,0}
\definecolor{yl}{RGB}{255,80,0}
\definecolor{mycy2}{RGB}{255,51,255}

\makeatletter
\newcommand*{\RN}[1]{\expandafter\@slowromancap\romannumeral #1@}
\makeatother

\usepackage{lineno}

\allowdisplaybreaks
 
\begin{document}

\date{}
\title{The Energy Complexity of Las Vegas Leader Election}

\author{
Yi-Jun Chang 
\\
\normalsize National University of Singapore \\
\normalsize \texttt{cyijun@nus.edu.sg} \\
\and
Shunhua Jiang \\
\normalsize Columbia University\\
\normalsize \texttt{sj3005@columbia.edu} \\
}

\begin{titlepage}
  \maketitle
  \thispagestyle{empty}

\begin{abstract}
We consider the  time (number of communication rounds) and energy (number of non-idle communication rounds per device) complexities of randomized leader election in a multiple-access channel, where the number of devices $n \geq 2$ is \emph{unknown}. 
It is well-known that for  polynomial-time randomized leader election algorithms with success probability $1 - 1/\poly(n)$, the optimal energy complexity is
$\Theta(\log \log^\ast n)$ if receivers can detect collisions, and it is $\Theta(\log^\ast n)$ otherwise.

Without collision detection,
all existing randomized leader election algorithms using $o(\log \log n)$ energy  are \emph{Monte Carlo} in that they might fail with some small probability, and they might consume unbounded energy and never halt when they fail. Although the optimal energy complexity of leader election \emph{appears} to have been settled,
it is still an intriguing open question whether it is possible to attain the optimal $O(\log^\ast n)$ energy complexity by an efficient \emph{Las Vegas} algorithm that never fails.
In this paper we address this fundamental question. 

\begin{description}
    \item[A separation between Monte Carlo and Las Vegas algorithms:] 
Without collision detection,
we prove that any Las Vegas leader election algorithm $\mathcal{A}$ with \emph{finite} expected time complexity must use $\Omega(\log \log n)$ energy, establishing a large separation between Monte Carlo and Las Vegas algorithms.
Our lower bound is \emph{tight}, matching the energy complexity of an existing leader election algorithm that finishes in $O(\log n)$ time and uses $O(\log \log n)$ energy in expectation.
    \item[An exponential improvement with sender collision detection:] 
In the setting where transmitters can detect collisions, we design a new leader election algorithm  that finishes in $O(\log^{1+\epsilon} n)$ time and uses $O(\epsilon^{-1} \log \log \log n)$ energy in expectation, showing that sender collision detection helps improve the energy complexity \emph{exponentially}. Before this  work, it was only known that sender collision detection is  helpful for \emph{deterministic} leader election.

    \item[An optimal deterministic leader election algorithm:] As a side result, via derandomization, we show a new deterministic leader election algorithm that takes $O\left(n \log \frac{N}{n}\right)$ time and $O\left(\log \frac{N}{n}\right)$ energy to elect a leader from $n$ devices, where each device has a unique identifier in $[N]$. The algorithm is \emph{simultaneously} time-optimal and energy-optimal, matching  existing $\Omega\left(n \log \frac{N}{n}\right)$ time lower bound and $\Omega\left(\log \frac{N}{n}\right)$ energy lower bound.
\end{description}

\end{abstract}

\end{titlepage}

{
\pagenumbering{gobble}
\hypersetup{linkcolor=black}
\doublespacing
\tableofcontents
\singlespacing
}

\newpage
\pagenumbering{arabic}

\section{Introduction}

\emph{Leader election} is one of the most central problems of distributed computing. In a network of an unknown number $n$ of devices communicating via a shared communication channel, the goal of leader election is to  have exactly one device in the network identify itself as the \emph{leader}, and all other devices identify themselves as \emph{non-leaders}. 

Leader election has a wide range of applications, as it captures the classic \emph{contention resolution} problem, where several processors need temporary and exclusive access to a shared resource.
Leader election is also used to solve the \emph{wake-up} problem~\cite{Gasieniec2001wakeup,Newport14}, whose the goal is to wake-up all processors in a completely connected broadcast system, in which an unknown number of processors are awake spontaneously and they have to wake-up the remaining sleeping processors.

We focus on \emph{single-hop} networks (all devices communicating via a shared communication channel) in the \emph{static} scenario (all devices start at the same time). 
Leader election protocols in single-hop networks are used as communication primitives in algorithms for more sophisticated distributed tasks in multi-hop  networks~\cite{bar1992time,Chang18broadcast}. Leader election protocols in the static setting  are  useful building blocks in the design of  contention resolution protocols in the \emph{dynamic} setting where the devices have different starting time~\cite{BenderFGY16,BenderKPY16} by batch processing.

\subsection{The Multiple-access Channel Model}

In our model, an unknown number $n$ of devices connect to a \emph{multiple-access channel}. The communication proceeds in synchronous rounds and all devices have an agreed-upon time zero.
In each communication round, a device may choose to  transmit a message, listen to the channel, or stay idle.

If more than one device simultaneously transmit a message in a round, then a \emph{collision} occurs. Listeners only receive messages from collision-free transmissions. There are four variants~\cite{ChangKPWZ17} of the model based on the collision detection ability of transmitters (distinguishing between successful transmission and  collision) and listeners (distinguishing between silence and collision).

\begin{description}
\item[$\strongcd$.] Transmitters and listeners receive one of the three  feedback:
(i)~silence, if zero devices transmit, 
(ii)~collision, if at least two devices transmit, or
(iii)~a message $m$, if exactly one device transmits.

\item[$\sendercd$.] Transmitters and listeners receive one of the two   feedback:
(i)~silence, if zero or at least two devices transmit, or
(ii)~a message $m$, if exactly one device transmits.

\item[$\receivercd$.] Transmitters receive no feedback.  Listeners receive one of the three  feedback:
(i)~silence, if zero devices transmit, 
(ii)~collision, if at least two devices transmit, or
(iii)~a message $m$, if exactly one device transmits.

\item[$\nocd$.] Transmitters receive no feedback. Listeners receive one of the two   feedback:
(i)~silence, if zero or at least two devices transmit, or
(ii)~a message $m$, if exactly one device transmits.
\end{description}

We distinguish between randomized and  deterministic models. In the randomized setting,  the devices are \emph{anonymous} in that they do not have unique identifiers and run the same algorithm, but they may break symmetry using their private random bits. In the deterministic setting,  each device is initially equipped with a unique identifier from an ID space $[N]$, where $N$ is \emph{global knowledge}. 
 Unless otherwise stated, we assume that the number of devices $n \geq 2$ is \emph{unknown}.

The goal of leader election is to  have exactly one device in the network identify itself as the {leader}, and all other devices identify themselves as {non-leaders}. We require that the communication protocol ends when the leader
sends a message while every non-leader listens to the channel, so all devices terminate in the same round.

\paragraph{Complexity measures.} 

Traditionally, the leader election problem has been studied from the context of optimizing the \emph{time complexity}, which is defined as the number of communication rounds needed to solve the problem.  More recently, there has been a growing interest~\cite{BenderKPY16,Chang2021detLE,ChangKPWZ17,caragiannis2005basic,lavault2007quasi,JurdzinskiKZ02cocoon,JurdzinskiKZ02podc,nakano2000randomized} in the  \emph{energy complexity} of leader election, which is defined as the maximum number of \emph{non-idle} rounds per device, over all devices. That is, each transmitting or listening round costs one unit of energy.  The study of energy complexity is motivated by the fact that many small mobile battery-powered devices are operated under a limited energy constraint. These devices may save energy by turning off their transceiver and entering a low-power sleep mode. As a large fraction of energy consumption of these devices are often spent on sending and receiving packets, the energy complexity of an algorithm approximates the actual energy usage of a device.
In applied research, \emph{idle listening} (transceiver is active but no data is received) has been identified as a major source of energy inefficiency in \emph{wireless sensor networks}, and there is a large body of work studying strategies for minimizing the number of transmission and idle listening~\cite{sadler2005fundamentals,miller2005mac,zhang2012mili,wang2012energy}.

\subsection{Prior Work}

For the time complexity of leader election,  Willard~\cite{Willard86} showed that expected $\Theta(\log\log n)$ time is necessary and sufficient
for leader election in $\receivercd$. More generally, Nakano and Olariu~\cite{nakano2002uniform} showed that the optimal time complexity of leader election in $\receivercd$ is $\Theta(\log\log n + \log f^{-1})$ if the maximum allowed failure probability is $f$. 

For the case an upper bound $\nm \geq n$ on the unknown network size $n$ is \emph{known} to all devices, it is well-known that  leader election can be solved using the \emph{decay} algorithm of Bar-Yehuda, Goldreich, and  Itai~\cite{bar1992time} in worst-case $O(\log  \nm \log f^{-1})$ time with success probability $1-f$~\cite{bar1992time,Gasieniec2001wakeup,JurdzinskiS02} in $\nocd$. The algorithm simply tries the transmission probability $2^{-i}$ for $O(f^{-1})$ times, for all  integers $1 \leq i  \leq \log \nm$ until a successful transmission occurs.
On the lower bound side, Jurdzi\'{n}ski and  Stachowiak~\cite{JurdzinskiS02} showed an $\Omega\left(\frac{\log  \nm \log f^{-1}}{\log \log \nm + \log \log f^{-1}}\right)$ time lower bound.
Later, Farach-Colton,   Fernandes, and   Mosteiro~\cite{Farach-ColtonFM06} showed the tight $\Omega(\log  \nm \log f^{-1})$ time lower bound for  \emph{oblivious} algorithms, in which there is a fixed sequence of transmission probabilities $(p_1, p_2, \ldots)$ such that if there has been no collision-free transmission, then  all devices transmit with the same probability $p_{i}$ in the $i$th round, using fresh randomness independently. 
Based on a technique of Alon, Bar-Noy, Linial, and Peleg~\cite{ablp91}, Newport~\cite{Newport14} showed an $\Omega(\log^2  \nm)$ time lower bound for the case $f = 1/ \poly(\nm)$ that applied to all algorithms. Very recently, the time complexity of leader election in which the algorithm is provided an arbitrary distribution of the possible network sizes $n$ was studied in~\cite{Gilbert2021}.


For the energy complexity of leader election algorithms, 
Lavault, Marckert, and Ravelomanana~\cite{lavault2007quasi} designed a leader election algorithm that finishes in expected $O(\log n)$ time  and uses expected $O(\log \log n)$ energy  in $\nocd$. Its expected time complexity matches the  $\Omega(\log n)$ lower bound for expected time of Newport~\cite{Newport14}. Subsequently, the algorithm of~\cite{lavault2007quasi} was applied to finding an estimate $\tilde{n}$ of $n$~\cite{vlady2016time}.

After a sequence of research~\cite{BenderKPY16,ChangKPWZ17,caragiannis2005basic,JurdzinskiKZ02podc,JurdzinskiKZ02cocoon}, it is now known~\cite{ChangKPWZ17} that for  polynomial-time randomized leader election algorithms with success probability $1 - 1/\poly(n)$, the optimal energy complexity is
$\Theta(\log \log^\ast n)$ if listeners can detect collisions ($\strongcd$ and $\receivercd$), and it is $\Theta(\log^\ast n)$ otherwise  ($\sendercd$ and $\nocd$).  

The energy complexity has also been studied in multi-hop networks~\cite{augustine2022distributed,berenbrink2009energy,barenboim_et_al:LIPIcs.DISC.2021.10,Chang18broadcast,Chang20bfs,ChatterjeeGP20,dani2021wake,dufoulon2022sleeping,gasieniec2007energy,Klonowski2018}. Optimization problems related to energy efficiency in  multi-hop networks were considered in~\cite{ambuhl2005optimal,Ambuhl2008,kirousis2000power}.

\subsection{Monte Carlo and Las Vegas Algorithms}\label{sect-two-types-rand-alg}

Although the optimal energy complexity of leader election \emph{appears} to have been settled due to the work of~\cite{ChangKPWZ17}, we observe that several existing randomized leader election protocols, including the ones in~\cite{ChangKPWZ17}, are \emph{Monte Carlo} in that they might fail with some small probability, and they might consume unbounded energy and never halt when they fail. 
In particular, without collision detection ($\nocd$), all existing randomized leader election algorithms~\cite{ChangKPWZ17,JurdzinskiKZ02cocoon} using $o(\log \log n)$ energy have this issue, and it is not even known if these algorithms have \emph{finite} expected time complexity.

It remains as an intriguing open question whether it is possible to attain the optimal $O(\log^\ast n)$ energy complexity proved in~\cite{ChangKPWZ17} by an efficient \emph{Las Vegas} algorithm that never fail.

It is tempting to guess that we might be able to transform existing Monte Carlo leader election algorithms into  Las Vegas algorithms  without worsening the asymptotic time and energy complexities by too much, but designing such a  transformation is actually very challenging when the number of devices $n$ is \emph{unknown}.

To explain the issue, consider the following simple  Monte Carlo leader election protocol that finishes in $O(\log^2 n)$ time with probability $1 - 1/\poly(n)$ when it is run on a network of $n$ devices. For $i = 1, 2, \ldots$,  the $i$th iteration consists of $C \cdot i$ rounds, where $C > 0$ is some constant. In each round each device $v$ transmits with probability  $2^{-i}$. All devices that are not transmitters listen to the channel, so a leader is elected once the number of transmitting devices is exactly one in a round.

Let $i^\ast = \lfloor \log n \rfloor$. If the number of devices is $n$, then the success probability in each round in the $i^\ast$th iteration is $\Omega(1)$, implying that a leader is elected by the $i^\ast$th iteration with probability $1 - n^{-\Omega(C)}$.
Hence the algorithm finishes in $\sum_{j=1}^{i^\ast} C \cdot i = O(\log^2 n)$ time with probability $1 - n^{-\Omega(C)}$.
The \emph{expected} time complexity of this protocol is however \emph{infinite} because in an extremely unlucky event that a leader is not elected within the first $O(\log n)$ iterations,  with high probability the protocol will run forever.

A natural attempt to fix the issue of infinite expected time complexity is to restart the protocol when it fails, but this strategy does not work as there is no mechanism for a device to detect that the algorithm has already failed! Since the number of devices $n$ is unknown, we are not able to set a time limit $T(n)$ and restart the protocol once  the number of rounds exceeds  $T(n)$ does not work, as the devices cannot calculate $T(n)$ if $n$ is not known. 

The issue is even more serious if there is no collision detection.  In the $\nocd$ model, each transmitter does not know if the message is successfully transmitted and each listener cannot distinguish between collision and silence, so there is no way for a device to learn anything about the number of devices $n$ given that no successful transmission (the number of transmitters is exactly one and the number of listeners is at least one) occurs.

\subsection{New Results}

The main objective of this paper is to better understand the strange gap between Monte Carlo and Las Vegas complexities of leader election discussed above. We focus on the following fundamental question: Is it possible to attain the optimal $O(\log^\ast n)$ energy bound proved in~\cite{ChangKPWZ17} by an efficient \emph{Las Vegas} algorithm in the $\nocd$ model? 

\paragraph{A separation between Monte Carlo and Las Vegas algorithms.}  
Surprisingly, we show that for any leader election algorithm $\mathcal{A}$ with a \emph{finite} expected time complexity in the $\nocd$ model, it is necessary that $\mathcal{A}$ uses $\Omega(\log \log n)$ energy, establishing a large separation between Monte Carlo and Las Vegas algorithms and answering the above question in the negative.

\begin{restatable}[Energy lower bound for Las Vegas algorithms]{theorem}{thmlbfinite}
\label{thm:energy_lb_finite_time}
Let $\mathcal{A}$ be a randomized leader election algorithm in the $\nocd$ model. We write $T(n)$ and $E(n)$ to denote the expected time and energy complexities of $\mathcal{A}$.
Suppose there is some integer $n^\ast \geq 2$ such the expected time $T(n^\ast)$ of the algorithm $\mathcal{A}$ when running on $n = n^\ast$ devices is finite. Then there exist infinite number of network sizes $n$ such that  $E(n) = \Omega(\log \log n)$.
\end{restatable}

Our lower bound is very strong in that the energy lower bound $\Omega(\log \log n)$ holds even if there is just \emph{one} network size $n^\ast$ such that the algorithm $\mathcal{A}$ has \emph{finite} expected time complexity when it is run on a network of $n^\ast$  devices.
Even allowing exponential time, our lower bound still rules out the possibility of having a Las Vegas leader election algorithm in $\nocd$ that uses $o(\log \log n)$ energy. 

Our lower bound is \emph{tight} in that it matches the energy complexity of the existing Las Vegas leader election algorithm of~\cite{lavault2007quasi}: There is a leader election algorithm that finishes in time $O(\log n)$ and energy $O(\log \log n)$ in expectation in the $\nocd$ model. As the expected time complexity $O(\log n)$ is already optimal due to the  lower bound  $\Omega(\log n)$ in~\cite{Newport14}, our result implies that the algorithm of~\cite{lavault2007quasi} is simultaneously time-optimal and energy-optimal.

\paragraph{An exponential improvement with sender collision detection.}  We design a new leader election algorithm in the $\sendercd$ model that finishes in $O(\log^{1+\epsilon} n)$ time and uses $O(\epsilon^{-1} \log \log \log n)$ energy in expectation, giving an \emph{exponential improvement} over the previous Las Vegas algorithm of~\cite{lavault2007quasi}, at the cost of slightly increasing the time complexity from $O(\log n)$ to $O(\log^{1+\epsilon} n)$.

\begin{restatable}[An exponential improvement in energy complexity]{theorem}{thmcd}
\label{thm:algo_sendercd}
For any $0 < \epsilon < 1$, there is an algorithm $\mathcal{A}$ in the $\sendercd$ model that elects a leader in expected $O(\log^{1+\epsilon} n)$ time and using expected $O(\epsilon^{-1} \log \log \log n)$ energy.
\end{restatable}

A fundamental problem in the study of multiple-access channels is to determine the value of collision detection.
It is well-known that  the ability for listeners to detect collision is very helpful in the design of randomized leader election algorithm in that the ability to distinguish between collision and silence allows the devices to perform an exponential search to estimate the network size $n$ efficiently~\cite{Willard86}.

Prior to this work, existing results suggested that the ability for transmitters to detect collision  does not seem to help in the randomized setting. Indeed, for polynomial-time Monte Carlo leader election algorithms, it was shown in~\cite{ChangKPWZ17} that $\Theta(\log^\ast n)$ is a tight energy bound in $\nocd$ and $\sendercd$ and  $\Theta(\log \log^\ast n)$ is a tight energy bound in $\receivercd$ and $\strongcd$, so it appears that the ability for transmitters to detect collision does not matter.

Our result shows that the  ability for transmitters to detect collision  helps improve the energy complexity {exponentially} for Las Vegas algorithms, giving the first example showing that the ability for transmitters to detect collision is valuable in the design of randomized algorithms.

We summarize our results for Las Vegas leader election algorithms in \cref{tab:expected}.

\begin{table*}[!ht]
\centering
\begin{tabular}{|l |l | l | l | l | l|}
\multicolumn{1}{l}{\bf Model} & 
\multicolumn{1}{l}{\bf Time} & 
\multicolumn{1}{l}{\bf Energy} & 
\multicolumn{1}{l}{\bf Type} &
\multicolumn{1}{l}{\bf Reference}
\\ \hline
\multirow{2}{*}{$\nocd$/$\sendercd$} & $n^{o(1)}$ & $ O(\log^* n)$  & Monte Carlo & \multirow{2}{*}{\cite{ChangKPWZ17}}
\\ \cline{2-4}
& $n^{O(1)}$ & $ \Omega(\log^* n)$  & Monte Carlo & 
\\ \hline
$\nocd$/$\sendercd$ & $O(\log n)$ & $O(\log \log n)$ & Las Vegas & \cite{lavault2007quasi}
\\ \hline 
$\nocd$ & $\Omega(\log n)$ & any & Las Vegas & \cite{Newport14}
\\ \hline 
$\nocd$ & any & $\Omega(\log \log n)$ & Las Vegas & \cref{thm:energy_lb_finite_time}
\\ \hline
$\sendercd$ & $O(\log^{1+\epsilon} n)$ & $O(\epsilon^{-1} \log \log \log n)$ & Las Vegas & \cref{thm:algo_sendercd}
\\ \hline
    \end{tabular}
    \caption{Old and new results on Las Vegas leader election algorithms. Here ``Monte Carlo'' indicates that the time and energy bounds hold with probability $1 - 1/\poly(n)$ and ``Las Vegas'' indicates that the time and energy bounds hold in expectation. }
    \label{tab:expected}
\end{table*}

\paragraph{An optimal deterministic leader election algorithm.} Recently, a systematic study of time-energy tradeoffs for deterministic leader election was done in~\cite{Chang2021detLE}. Due to the result of~\cite{Chang2021detLE},  for the three models $\strongcd$, $\receivercd$, and $\sendercd$, tight or nearly tight time and energy bounds in terms of the number of devices $n$ and  the size of ID space $N$ were known for deterministic leader election. 

The last missing piece in the puzzle is the $\nocd$ model, where the current best deterministic $\nocd$ algorithm takes $O(N)$ time and $O\left(\log \frac{N}{n}\right)$ energy~\cite{Chang2021detLE} and the current best lower bounds are $\Omega\left(n \log \frac{N}{n}\right)$ time~\cite{clementi2003distributed} and $\Omega\left(\log \frac{N}{n}\right)$ energy~\cite{Chang2021detLE}. 

By a derandomization of a subroutine that we use in our randomized Las Vegas algorithms, we show an optimal deterministic $\nocd$ algorithm that takes $O\left(n \log \frac{N}{n}\right)$ time~\cite{clementi2003distributed} and $O\left(\log \frac{N}{n}\right)$ energy, settling the optimal complexities of leader election in the deterministic $\nocd$ model.

\begin{restatable}[Optimal deterministic leader election]{theorem}{thmoptdet}
\label{thm:optimal_algo_no_cd}
Suppose that the size  $N$ of the ID space $[N]$  and an estimate $\tilde{n}$ of the  number of devices $n$ such that  $\tilde{n}/2 < n \leq \tilde{n}$ are both known to all devices. There is a deterministic leader election algorithm in the $\nocd$ model with time complexity $T = O\left(n \log \frac{N}{n}\right)$ and energy complexity $E  = O\left(\log \frac{N}{n}\right)$.
\end{restatable}

Our deterministic algorithm requires that  $n$ is known or a constant-factor approximation of $n$ is given.
See \cref{tab:deterministic} for a summary of results on deterministic algorithms in $\nocd$.


\begin{table*}[!ht]
\centering
\renewcommand{\arraystretch}{1.2}{\begin{tabular}{|l |l | l | l | l | l|}
\multicolumn{1}{l}{\bf Model} & 
\multicolumn{1}{l}{\bf Time} & 
\multicolumn{1}{l}{\bf Energy} &
\multicolumn{1}{l}{\bf Reference}
\\ \hline
$\nocd$ & $O(N)$ & $O(\log \frac{N}{n})$ & \cite{Chang2021detLE}
\\ \hline
$\nocd$ & $\Omega(n \log \frac{N}{n})$ & any & \cite{clementi2003distributed}
\\ \hline 
$\nocd$ & any & $\Omega(\log \frac{N}{n})$ & \cite{Chang2021detLE}
\\ \hline 
$\nocd$ & $O(n \log \frac{N}{n})$ & $O(\log \frac{N}{n})$ & \cref{thm:optimal_algo_no_cd}
\\ \hline
    \end{tabular}
    }
    \caption{Old and new results on deterministic leader election algorithms in $\nocd$, where $N$ indicates the size of the ID space and $n$ indicates the \emph{known} number of devices.}
    \label{tab:deterministic}
\end{table*}

\subsection{Technical Overview}


In this section we overview of the key ideas behind the proofs of our results.

\paragraph{The $\Omega(\log \log n)$ energy lower bound.}  
Recall from the discussion in \cref{sect-two-types-rand-alg} that the main source of difficulty of transforming a Monte Carlo algorithm
 a  into a Las Vegas one in the $\nocd$ model is that a device is unable to obtain any information from listening to the channel if no successful transmission occurs, as the feedback from the channel is always silence.  The behavior of such a device  depends only on its own private randomness, as it does not receive any other information. In particular, such a device cannot learn anything about $n$.
A key idea behind our $\Omega(\log \log n)$ energy lower bound in \cref{thm:energy_lb_finite_time} is to make use of this observation.
The proof of \cref{thm:energy_lb_finite_time} combines the following three ingredients. 
\begin{itemize}
    \item It was shown in~\cite{Newport14} that with constant probability there is no collision-free transmission in the first $t = O(\log n)$ rounds. This means that to prove the $\Omega(\log \log n)$ energy lower bound, it suffices to show that under the condition that the channel feedback is always silence, the energy cost in the first $t$ rounds is $\Omega(\log t) = \Omega(\log \log n)$ in expectation.
    \item The assumption that the given algorithm $\mathcal{A}$ has finite expected time complexity for some fixed $n=n^\ast$ implies that there are infinitely many $t$ such that the probability that the algorithm does not finish by time $t$ less than $f = 1/t$.
    \item By a derandomization, we may transform the $\Omega\left(\log \frac{N}{n}\right)$ deterministic energy lower bound in~\cite{Chang2021detLE}  into a randomized lower bound $\Omega\left(n^{-1} \log f^{-1}\right)$ for $f < 1 / \binom{N}{n}$.
\end{itemize}

Setting $f = 1/t$ and $n = n^\ast = \Theta(1)$ in the randomized lower bound $\Omega\left(n^{-1} \log f^{-1}\right) = \Omega(\log t)$, we obtain that the energy cost in the first $t$ rounds is $\Omega(\log t)$ when we run the algorithm $\mathcal{A}$ under the condition that the channel feedback is always silence.
Combining this with the lower bound of~\cite{Newport14}, we obtain the desired $\Omega(\log \log n)$-energy lower bound.
This lower bound argument still works even if the underlying network size $n$ is not $n^\ast$ because a device does not learn anything about $n$ if the channel feedback is always silence.



\paragraph{Las Vegas Leader election algorithms.} To prove \cref{thm:algo_sendercd}, we will first 
 design a basic subroutine that achieves the following. Given a network size estimate $\tilde{n}$, the subroutine elects a leader in $O(\log f^{-1})$ time with $O(\tilde{n}^{-1} \log f^{-1})$ energy in $\nocd$, if the number of devices $n$ satisfies $\tilde{n}/2 < n \leq \tilde{n}$.
 
 Using this subroutine, we may re-establish the result of~\cite{lavault2007quasi} that leader election can be solved in expected $O(\log n)$ time and using expected $O(\log \log n)$ energy.  The leader election  proceeds in iterations. During the $i$th iteration, we run our basic subroutine for all $\tilde{n} = 2^{1}, 2^2, 2^3,\ldots, 2^{2^i - 1}$ with $f = 1/4$, and then the algorithm terminates once a leader is elected. Intuitively, what the algorithm does in iteration $i$
 is that it goes over all network size estimates $\tilde{n}$ from $2^1$ to $2^{2^i - 1}$ and spend $O(1)$ time for each $\tilde{n}$. Once the number of devices $n$ satisfies $\tilde{n}/2 < n \leq \tilde{n}$, then a leader is elected with constant probability.
 
The main source of energy inefficiency of the above algorithm is the high energy cost of the basic subroutine with small $\tilde{n}$-values. 
In particular, if $\tilde{n} = O(1)$, Then the energy cost of 
 achieve a success probability of $1 - f$ is $O(\log f^{-1})$ using our basic subroutine.
 
A key idea behind the proof of \cref{thm:algo_sendercd} is an observation that this energy complexity $O(\log f^{-1})$ can be improved exponentially in the $\sendercd$ model  at the cost of increasing the time complexity. 
To see this, first consider the case when the number of devices is $n=2$. We allocate an ID space of size $N = \lceil f^{-1} \rceil$, let each device choose an ID uniformly at random from $[N]$, and run the $O(N)$-time and $O(\log \log N)$-energy deterministic $\sendercd$ leader election algorithm of~\cite{ChangKPWZ17}. As long as the two devices select different IDs, the algorithm of~\cite{ChangKPWZ17} successfully elects a leader.
Therefore, with probability at least $1 - f$, this procedure elects a leader in
$O(N) = O(f^{-1})$ time and uses $O(\log \log N) = O(\log \log f^{-1})$ energy. 

By switching to a more energy-efficient $\sendercd$ algorithm when dealing with small $\tilde{n}$-values, we are able to achieve an exponential improvement in the energy complexity.
It is crucial that the increase in the time complexity is not too much when we switch from our basic subroutine to some other $\sendercd$ algorithm. To ensure that our final algorithm has a finite expected time complexity, the probability that the algorithm does not terminate by time $t$ has to be $o(t^{-1})$ for all but a finite number of rounds $t$. To put it another way, for a given failure probability parameter $f$, the maximum allowed time complexity will be $o(f^{-1})$, meaning that the  simple $\sendercd$ algorithm for $n=2$ presented above is not suitable for our purpose.

To ensure that the time complexity is sufficiently small $o(f^{-1})$, we will employ a recent result~\cite{Chang2021detLE} on time-energy tradeoffs of deterministic leader election in the $\sendercd$ model.
By properly incorporating the $\sendercd$ algorithm of~\cite{Chang2021detLE} into our framework to take care of small $\tilde{n}$-values, we obtain a leader election algorithm $\mathcal{A}$ in the $\sendercd$ model that elects a leader in expected $O(\log^{1+\epsilon} n)$ time and using expected $O(\epsilon^{-1} \log \log \log n)$ energy, proving \cref{thm:algo_sendercd}.

 \paragraph{Derandomization.} The proof of \cref{thm:optimal_algo_no_cd} follows from a derandomization of our basic subroutine, which elects a leader in $O(\log f^{-1})$ time with $O(\tilde{n}^{-1} \log f^{-1})$ energy in $\nocd$. 
 In \cref{thm:optimal_algo_no_cd}, we consider the deterministic setting where the number of devices $n$ satisfies $\tilde{n}/2 < n \leq \tilde{n}$ and each device has a unique identifier in $[N]$.
 To derandomize our basic subroutine, we set \[f = 2^{-\Omega\left(\tilde{n} \log \frac{N}{\tilde{n}}\right)} <  \frac{1}{\sum_{n \in (\tilde{n}/2, \tilde{n}]} \binom{N}{n}}\] and apply a union bound over all size-$n$ subsets of $[N]$ for  $\tilde{n}/2 < n \leq \tilde{n}$. 
 Our choice of $f$ implies that the resulting deterministic algorithm has time complexity $O(\log f^{-1}) = O\left(\tilde{n} \log \frac{N}{{n}}\right)$ and energy complexity $O(\tilde{n}^{-1} \log f^{-1}) = O\left(\log \frac{N}{{n}}\right)$, as required.


\subsection{Organization}
In \cref{sect:finite-lower-bound}, we present an $\Omega(\log \log n)$ energy lower bound for Las Vegas leader election algorithms in  $\nocd$, proving \cref{thm:energy_lb_finite_time}.
In \cref{sect:basic}, we present our basic subroutine which we use in our randomized algorithms, and we derandomize it to obtain an optimal deterministic leader election algorithm in  $\nocd$, proving \cref{thm:optimal_algo_no_cd}.
In \cref{sect-rand-main}, we present a framework for designing leader election algorithms and use it to reprove the result of~\cite{lavault2007quasi} 
that in $\nocd$ a leader can be elected in expected $O(\log n)$ time and using expected $O(\log \log n)$ energy.
Finally, in \cref{sect:sender-main}, using our framework, we present a leader election algorithm in the $\sendercd$ model that finishes in expected $O(\log^{1+\epsilon} n)$ time and using expected $O(\epsilon^{-1} \log \log \log n)$ energy, proving \cref{thm:algo_sendercd}.

\section{A Tight Energy Lower Bound for Las Vegas Algorithms}\label{sect:finite-lower-bound}
In this section we prove \cref{thm:energy_lb_finite_time}, which gives an expected $\Omega(\log \log n)$ energy lower bound for any randomized algorithm $\mathcal{A}$ in the $\nocd$ model that works for an unknown number of devices $n$ and has a finite expected time complexity. 
In fact, we will prove the following stronger lower bound: With constant probability, the average energy cost among all $n$ devices is $\Omega(\log \log n)$ in an execution of $\mathcal{A}$ on $n$ devices.


\paragraph{Successful transmission.}
Instead of working with the leader election problem directly, we will consider the easier problem of having just one \emph{successful transmission}, which is defined as follows. We say that a successful transmission occurs in a round if there is exactly one device transmitting in this round and there is at least one device listening in this round. We allow some devices to be idle when a successful transmission occurs.
As we only consider the $\nocd$ model in this section, the feedback from the communication channel is always silence if there is no successful transmission, so every device does not receive external information. 

\paragraph{Description of an algorithm.}
Recall that in the deterministic setting, we assume each device $v$ has a unique identifier $\ID(v) \in [N]$, where the size of ID space $N$ is a global knowledge.
As we do not care about the behavior of the algorithm after the first successful transmission, we may assume that a deterministic algorithm 
$\mathcal{A}$ is specified by a mapping $\phi$ from the ID space $x \in [N]$ to an infinite sequence of actions $(a_{1}(x), a_2(x), \ldots)$, where each $a_i(x) \in \{\transmit, \listen, \idle\}$ specifies the action of a device $v$ with $\ID(v)=x$ in the $i$th round, assuming that the channel feedback is always silent whenever $v$ listens in the first $i-1$ rounds.
Similarly,  a randomized algorithm $\mathcal{A}$ is specified by a  distribution $\mathcal{D}$ of infinite sequences of actions $(a_{1}, a_2, \ldots)$. When a device $v$ runs a randomized algorithm $\mathcal{A}$, it uses its private random bits to sample $(a_{1}, a_2, \ldots) \sim \mathcal{D}$ to determine its action  $a_i \in \{\transmit, \listen, \idle\}$ in each round $i$ if the channel feedback is always silent whenever $v$ listens in the first $i-1$ rounds. Throughout this section, we use the above notation for describing a deterministic or a randomized algorithm $\mathcal{A}$.

\paragraph{Deterministic energy lower bound.}
Our proof of \cref{thm:energy_lb_finite_time} relies on transforming deterministic lower bounds into randomized lower bounds. 
We first prove an $\Omega(\log \frac{N}{n})$ energy lower bound for deterministic algorithms in $\nocd$. This lower bound is a slightly stronger version of \cite[Theorem~3]{Chang2021detLE} which considers the \emph{average} energy cost. 

\begin{lemma}
[{Generalization of \cite[Theorem~3]{Chang2021detLE}}]\label{lem:energy_lb_log_N_n}
Let $N$ be the size of the ID space, and let $2 \leq n \leq N/2$ be the number of devices. 
We allow  both $N$ and $n$ to be global knowledge.
Let $\mathcal{A}$ be any deterministic algorithm in the $\nocd$ model that guarantees a successful transmission in the first $t$ rounds for any choice of the size-$n$ subset $V \subseteq [N]$ of devices.
Let \[k_j = \left| \; \left\{ a_i(j) \ : \ \left(i\in[t]\right)  \wedge  \left( a_i(j) \neq \idle\right) \right\} \; \right|\] be the energy cost of a device with ID $j$ in the first $t$ rounds assuming that the channel feedback is always silent. Then we have
\begin{align*}
    \frac{1}{N} \sum_{j=1}^N k_j = \Omega\left(\log \frac{N}{n}\right).
\end{align*}
\end{lemma}
\begin{proof}
Re-order the $N$ values $\{k_1, k_2, \ldots, k_N\}$ such that $k_1 \leq k_2 \leq \cdots \leq k_N$. We will prove that $k_{N/2} = \Omega(\log \frac{N}{n})$, so the average value satisfies $\frac{1}{N} \sum_{j=1}^N k_j \geq \frac{1}{N} \cdot (\frac{N}{2} \cdot k_{N/2}) = \Omega(\log \frac{N}{n})$, as required.

We consider a random sequence $\{b_i\}_{i=1}^t$ where each $b_i$ is uniformly randomly sampled from $\{\listen, \transmit\}$. We say $\{b_i\}_{i=1}^t$ matches a sequence $\{a_i(j)\}_{i=1}^t$ if for any $i$, either $a_i(j) = \idle$, or $a_i(j) = b_i$. For each $j \in [\frac{N}{2}]$, since there are $k_j$ $\listen$ or $\transmit$ actions in the sequence $\{a_i(j)\}_{i=1}^t$, it is easy to see that
\[
\Pr[\{b_i\}_{i=1}^t \text{~matches~} \{a_i(j)\}_{i=1}^t] = \frac{1}{2^{k_j}} \geq \frac{1}{2^{k_{N/2}}}.
\]
Thus in expectation $\{b_i\}_{i=1}^t$ matches $N / 2^{k_{N/2}+1}$ number of action sequences in $[\frac{N}{2}]$. This means there must exist some $\{b_i\}_{i=1}^t$ that matches at least $N / 2^{k_{N/2}+1}$ number of action sequences in $[\frac{N}{2}]$. Let $V \subseteq [\frac{N}{2}]$ denote the set of devices that this $\{b_i\}_{i=1}^t$ matches with, and we have $|V| \geq N / 2^{k_{N/2}+1}$.

The algorithm $\mathcal{A}$ cannot be correct on any set $V' \subseteq V$, because in any round all the devices in $V'$ either all perform actions in $\{\listen, \idle\}$, or all perform actions in $\{\transmit, \idle\}$, so there does not exist a round where exactly one device transmits and at least one device listens. Since the algorithm $\mathcal{A}$ is correct on all sets of size $n$, we must have $n > |V| \geq N / 2^{k_{N/2}+1}$, which gives $k_{N/2} = \Omega(\log \frac{N}{n})$.
\end{proof}

\paragraph{Randomized energy lower bound.}
Next, we prove the following energy lower bound for algorithms with success probability $f$ using a reduction from the previous $\Omega(\log \frac{N}{n})$ energy lower bound for deterministic algorithms in $\nocd$.

Note that if $f \geq (8e)^{-n}$, then $\frac{1}{n} \log \frac{1}{f} = \log(8e) \leq 5$, so $\Omega(n^{-1} \log f^{-1})$ becomes $\Omega(1)$, which is a trivial lower bound for leader election as well as many non-trivial tasks. In this sense the assumption $f < (8e)^{-n}$ in the following lemma is justified.


\begin{lemma}[Energy lower bound for algorithms with error probability $f$]\label{lem:energy_lb_f_prob}
Let $0 < f < (8e)^{-n}$. Let $\mathcal{A}$ be a randomized algorithm in the $\nocd$ model that satisfies the following.  When $\mathcal{A}$ is executed on a network of $n$ devices,  with probability at least $1 - f$, a successful transmission occurs in the first $t$ rounds. We allow both parameters $n$ and $f$ to be global knowledge.

Then there exists an integer $N = \Theta\left(n \cdot (1/f)^{1/n}\right)$ such that the following is satisfied. Define $k_1, k_2, \ldots, k_N$ as independent random variables
\[k_j = \left| \; \left\{ a_i(j) \ : \ \left(i\in[t]\right)  \wedge  \left( a_i(j) \neq \idle\right) \right\} \; \right|,  \text{~where~} (a_1(j), a_2(j), \ldots) \sim \mathcal{D},\] representing the energy cost of $\mathcal{A}$ for a device $v$ with $\ID(v) = j$ in the first $t$ rounds assuming that the channel feedback is always silence. Then we have \[\Pr\left[\frac{1}{N} \sum_{j=1}^N k_j = \Omega\left(\frac{1}{n} \log \frac{1}{f}\right)\right] \geq \frac{3}{4}.\]
As a result, the expected value of each $k_j$ is $\Omega(n^{-1} \log f^{-1})$.
\end{lemma}
\begin{proof}
Let $N$ be the largest integer that satisfies $\binom{N}{n} < \frac{1}{4f}$. Since $f < (8e)^{-n}$ and $\binom{2n}{n} \leq (2e)^n$, we must have $N \geq 2n$. By definition we have 
\[
\frac{1}{4f} \leq \binom{N+1}{n} = \binom{N}{n} \cdot \frac{N+1}{N-n+1} \leq 2 \cdot \binom{N}{n}.
\]
Thus we have $\frac{1}{f} = \Theta\left(\binom{N}{n}\right)$, and hence $\log \frac{1}{f} = \Theta\left(n \cdot \log \frac{N}{n}\right)$ and $N = \Theta\left(n \cdot (1/f)^{1/n}\right)$.

Consider the ID space $[N]$. For each $j \in [N]$, we fix the private random bits for the device $v$ with $\ID(v)$ independently, and let 
\[(a_{1}(j), a_2(j), \ldots) \sim \mathcal{D}\]
be the actions of $v$ for each round $i$ when running $\mathcal{A}$ using its private random bits, assuming that the channel feedback is always silent whenever $v$ listens in the first $i-1$ rounds. We also define
\[k_j = \left| \; \left\{ a_i(j) \ : \ \left(i\in[t]\right)  \wedge  \left( a_i(j) \neq \idle\right) \right\} \; \right|.\] 


For a fixed subset $S$ of $n$ devices of $[N]$, the algorithm $\mathcal{A}$ executed on $S$ guarantees a successful transmission within the first $t$ rounds with probability at least $1 - f$. We define  $X$ to be the event that for all subsets $S$ of $n$ devices of $[N]$, $\mathcal{A}$ successfully elects a leader. Note that $X$ depends on the private random bits used in sampling $(a_{1}(j), a_2(j), \ldots) \sim \mathcal{D}$ for all $j \in [N]$. Using a union bound, and since $\binom{N}{n} < \frac{1}{4f}$,
\[
\Pr[X] \geq 1 - f \cdot \binom{N}{n} \geq \frac{3}{4}.
\]

If for some fixed random bits the event $X$ happens, then the algorithm $\mathcal{A}$ with those fixed random bits yields a deterministic algorithm for ID space $[N]$. The $\Omega(\log \frac{N}{n})$ energy lower bound of \cref{lem:energy_lb_log_N_n} then implies that the lower  bound $\frac{1}{N} \sum_{j=1}^N k_j = \Omega\left(\log \frac{N}{n}\right)$ holds whenever $X$ happens. Thus we have
\[
\Pr\left[\frac{1}{N} \sum_{j=1}^N k_j = \Omega\left(\frac{1}{n} \log \frac{1}{f}\right)\right] = \Pr\left[\frac{1}{N} \sum_{j=1}^N k_j = \Omega\left(\log \frac{N}{n}\right)\right] \geq \frac{3}{4}.\qedhere
\]
\end{proof}

In the following, we apply \cref{lem:energy_lb_f_prob} to analyze the energy complexity of any algorithm  $\mathcal{A}$ meeting the conditions specified in \cref{thm:energy_lb_finite_time}.

\begin{lemma}[$\Omega(\log t)$ energy in $t$ rounds]\label{cor:log_t_energy_t_rounds}
Let $\mathcal{A}$ be a randomized leader election algorithm in the $\nocd$ model that works for any unknown number of devices $n \geq 2$, and it has a finite expected time complexity for some $n = n^\ast$.

Then there exists an infinite set $S$ of positive integers such that for any $t \in S$, there exists an integer $N = \Theta\left(t^{1/n^\ast}\right)$ such that the following is satisfied. Define $k_1, k_2, \ldots, k_N$ as independent random variables
\[k_j = \left| \; \left\{ a_i(j) \ : \ \left(i\in[t]\right)  \wedge  \left( a_i(j) \neq \idle\right) \right\} \; \right|,  \text{~where~} (a_1(j), a_2(j), \ldots) \sim \mathcal{D},\] representing the energy cost of $\mathcal{A}$ for a device $v$ with $\ID(v) = j$ in the first $t$ rounds assuming that the channel feedback is always silent. Then we have \[\Pr\left[\frac{1}{N} \sum_{j=1}^N k_j = \Omega\left(\log t\right)\right] \geq \frac{3}{4}.\]
In particular, the expected value of each $k_j$ is $\Omega(\log t)$.
\end{lemma}
\begin{proof}
The expected time of $\mathcal{A}$ when running on $n^\ast$  devices is
\[
\sum_{t=1}^{\infty} \Pr[\mathcal{A} \text{~does not finish  within $t$ rounds when running on $n^\ast$ devices}].
\]
Since we assume this expected time is finite, and $\sum_{t=1}^{\infty} \frac{1}{t} = \infty$, there must exist an infinite set $S$ of positive integers such that for any $t \in S$, 
\[
\Pr[\mathcal{A} \text{~does not finish  within $t$ rounds when running on $n^\ast$ devices}] \leq \frac{1}{t}.
\]

Fix any $t \in S$.
If we run $\mathcal{A}$ on $n^\ast$ devices, 
then  $\mathcal{A}$ successfully elects a leader within the first $t$ rounds with probability at least $1 - \frac{1}{t}$. In particular, it guarantees a successful transmission within the first $t$ rounds with probability at least $1 - \frac{1}{t}$.

Therefore, applying  \cref{lem:energy_lb_f_prob} with $f = 1/t$ and $n = n^\ast$, and since $n^*$ is a fixed constant, there exists an integer $N = \Theta\left(n^\ast \cdot t^{1/n^\ast}\right) = \Theta\left(t^{1/n^\ast}\right)$ such that the conclusion of this lemma is satisfied. Although the criterion $0 < f < (8e)^{-n}$ needed for applying  \cref{lem:energy_lb_f_prob} might not be met when $t$ is small, the set $S$ excluding those small $t$ is still infinite.
\end{proof}

The last missing piece is the following lemma that shows an $\Omega(\log n)$ time lower bound when the algorithm succeeds with constant probability. The proof of this lemma   follows from the same proof argument of~\cite[Theorem 5.2]{Newport14}, which shows an $\Omega(\log n)$ \emph{expected} time lower bound, so here we only include a proof sketch. For the sake of completeness, we include a full proof in  \cref{app:missing-proof}. We note that the purpose of the  technical condition $\sqrt{\tilde{n}} \leq n \leq \tilde{n}$ in the lemma is to ensure that the $\Omega(\log \log n)$ energy lower bound in \cref{thm:energy_lb_finite_time} applies to infinitely many $n$.

\begin{lemma}[$\Omega(\log n)$ time for constant success probability \cite{Newport14}]\label{lem:logn_time_lb}
Let $\mathcal{A}$ be a randomized algorithm in the $\sendercd$ model such that for any known integer $\tilde{n} \geq 2$, for any unknown number of devices $\sqrt{\tilde{n}} \leq n \leq \tilde{n}$, the probability that a collision-free transmission occurs by time $T(\tilde{n})$ in an execution of $\mathcal{A}$ on $n$ devices is at least $1/4$. Then we have $T(\tilde{n}) = \Omega(\log \tilde{n})$.
\end{lemma}
\begin{proof}[Proof sketch] 
Following the terminology of \cite{Newport14}, for any two sets $F,H \subseteq [N]$, we say that $F$ hits $H$ if $|F \cap H| = 1$. In Theorem 3.1 of \cite{alon2014broadcast}, it was shown that given any integer $N \geq 2$, there exists a multiset $\mathcal{H}$ of subsets $H \subseteq [N]$ such that every subset $F \subseteq [N]$ hits at most $O\left(\frac{1}{\log N}\right)$ fraction of $\mathcal{H}$. In fact this result also holds with the additional requirement that each subset $H$ in the multiset $\mathcal{H}$ has size $\sqrt{N} \leq |H| \leq N$. This follows from a straightforward extension of the original proof of \cite{alon2014broadcast}, and for simplicity we omit it.

We follow the same proof strategy as Theorem 5.2 of \cite{Newport14} which proves an $\Omega(\log n)$ expected time lower bound using Theorem 3.1 of \cite{alon2014broadcast}. Construct the multiset $\mathcal{H}$ with $N = \tilde{n}$. For each set $H \subseteq [\tilde{n}]$ in the multiset $\mathcal{H}$, consider running the algorithm $\mathcal{A}$ on the devices in $H$. In each round, a collision-free transmission occurs if and only if the set of the transmitting devices $F \subseteq [\tilde{n}]$ hits the set $H$.  Since any set $F$ hits at most $O\left(\frac{1}{\log \tilde{n}}\right)$ fraction of $\mathcal{H}$, in order to achieve a success probability of $1/4$, $\mathcal{A}$ needs to run for at least $\Omega(\log \tilde{n})$ rounds.
\end{proof}

Now we are ready to prove \cref{thm:energy_lb_finite_time} by combining the $\Omega(\log n)$ time lower bound in \cref{lem:logn_time_lb} with the $\Omega(\log t)$ energy lower bound in \cref{cor:log_t_energy_t_rounds}. 


\thmlbfinite*

\begin{proof}
From \cref{lem:logn_time_lb} we have that there exists a constant $c > 0$ such that for each integer $\tilde{n} \geq 2$, there exists an integer $\sqrt{\tilde{n}} \leq n \leq \tilde{n}$  such that when running $\mathcal{A}$ on $n$ devices, we have
\begin{equation}\label{eq:energy_lb_finite_time_1}
    \Pr\left[\text{~no successful transmission within the first $c \log n$ rounds~}\right] \geq 3/4.
\end{equation}


Let $S$ be the infinite set obtained from applying \cref{cor:log_t_energy_t_rounds} with the algorithm $\mathcal{A}$. Consider any $t \in S$. 
Choose $\tilde{n} = 2^{2t/c}$, and then pick an integer $\sqrt{\tilde{n}} \leq n \leq \tilde{n}$  such that when running $\mathcal{A}$ on $n$ devices, Eq.~\eqref{eq:energy_lb_finite_time_1} holds. The existence of such an integer $\tilde{n}$ is guaranteed by \cref{lem:logn_time_lb}. 

Our choice of $n$ ensures that $t \leq c \log n$, so Eq.~\eqref{eq:energy_lb_finite_time_1} implies that 
when running $\mathcal{A}$ on $n$ devices, with probability at least $3/4$, no successful transmission occurs within the first $t = \Theta(\log n)$ rounds. 

Consider the integer integer $N = \Theta\left(t^{1/n^\ast}\right)$ in \cref{cor:log_t_energy_t_rounds}.
Let $k$ denote the average energy used by the first $N$ devices in $t$ rounds. \cref{cor:log_t_energy_t_rounds} implies that 
\begin{equation}\label{eq:energy_lb_finite_time_2}
    \Pr[k \geq \Omega(\log t)] \geq 3/4.
\end{equation}

Combining Eq.~\eqref{eq:energy_lb_finite_time_1} and \eqref{eq:energy_lb_finite_time_2} with a union bound, we obtain that with probability at least $1/2$, when running $\mathcal{A}$ on $n$ devices, there is at least one device that uses at least $\Omega(\log t) = \Omega(\log \log n)$ energy. Since there are infinitely many $t \in S$, we are able to select infinitely many $n$ such that the $\Omega(\log \log n)$ lower bound holds.

For the rest of the proof, we extend the above argument to show that in fact  with probability at least $1/2$, when running $\mathcal{A}$ on $n$ devices, the average energy cost per device is $\Omega(\log t) = \Omega(\log \log n)$, so the expected energy cost per device is indeed $\Omega(\log \log n)$.

Let $v_1, v_2, \ldots, v_n$ denote the $n$ devices. For each $j \in [n]$, let  $k_j$ be the random variable representing the energy cost of $\mathcal{A}$ for $v_j$ in the first $t$ rounds assuming that the channel feedback is always silent. It is clear that $k_1, k_2, \ldots, k_n$ are independent, and each $k_j$ depends only on the private random bits in $v_j$. It suffices to show that 
\[\Pr\left[\frac{1}{n} \sum_{j=1}^n k_j = \Omega\left(\log t\right)\right] \geq \frac{3}{4}.\]

To prove this bound, we partition $\{k_1, k_2, \ldots, k_n\}$ into $\lfloor{n/N}\rfloor$ disjoint groups of size exactly $N$ and at most one leftover group. For each group, with probability at least $3/4$, the average value is $\Omega(\log t)$. Applying a Chernoff bound over all $\lfloor{n/N}\rfloor$ disjoint groups of size $N$, we obtain that the average value of $\{k_1, k_2, \ldots, k_n\}$ is $\Omega(\log t)$ with probability $1 - e^{-\Omega(n/N)} = 1 - e^{-\Omega(n / \log n)} \gg 3/4$, as long as $n$ is sufficiently large, as $N = O\left(t^{1/n^\ast}\right) = o(\log n)$.
\end{proof}

\section{Leader Election with a Network Size Estimate}\label{sect:basic}

For upper bounds, we begin with the case where an estimate $\tilde{n}$ of the actual number of devices $n$ is given. Given a parameter $0 < f < 1$, our task is to elect a leader with probability $1 - f$ when the estimate $\tilde{n}$ satisfies $\tilde{n}/2 < n \leq \tilde{n}$. 

In this section, we design an algorithm solving this task with $O(\log f^{-1})$ time and $O(\tilde{n}^{-1} \log f^{-1})$ energy in the $\nocd$ model. In \cref{sect:multi} we extend this algorithm to deal with multiple pairs of $(\tilde{n}, f)$, and the resulting algorithm will later be used as a subroutine in our leader election algorithms for the more challenging scenario where the number of devices $n$ is completely unknown. In \cref{sect:derandomize}, we derandomize our algorithm to give an optimal deterministic leader election algorithm. 

\paragraph{Balls-into-bins.} We need the following balls-into-bins lemma. In this lemma we care about a subset of the bins, which we call ``good'' bins, and the rest of the bins are called ``bad'' bins. The goal of this lemma is to analyze the number of good bins that contain exactly one ball.

In \cref{lem:good_bins_one_ball}, the numbers $\alpha N$ and $\gamma N$ are not required to be integers, but $n = \beta N$ must be an integer, as it specifies the number of balls in a balls-into-bins experiment.

\begin{lemma}[Number of good bins with exactly one ball]\label{lem:good_bins_one_ball}
Let $N$ be an integer. Let $\alpha, \beta, \gamma \in (0, 1)$ be three numbers satisfying $\alpha > 2 \beta + \gamma$. Let $t$ be any parameter such that $t < \alpha - 2 \beta - \gamma$. There are $n = \beta N$ balls and $N$ bins. Among the bins there are at least $\alpha N$ good bins, and the rest of the bins are called bad bins.

Consider the following balls-into-bins experiment.
For each ball, there is an arbitrary subset of at least $(1-\gamma)N$ bins such that the ball is uniformly randomly thrown into one of the bins in this subset. 
With probability at least $1 - e^{-\frac{t^2 n}{2}}$, there exist at least $(\alpha - 2 \beta - \gamma - t )n$ good bins that contain exactly one ball.
\end{lemma}
\begin{proof}
We use a random variable $Y$ to denote the total number of good bins that contain exactly one ball.

We define $n$ random variables $X_1, X_2, \ldots, X_n \in \{-1,0,1\}$ as follows. Before throwing the $i$th ball into a random bin, we first reorder the bins in such a way that the good bins that are empty are in the front, so the first $(\alpha - \beta)N$ bins are always empty good bins. We define 
\[
    X_i = 
    \begin{cases}
    1 & \text{if the $i$th ball is thrown into the first $(\alpha - \beta)N$ bins,} \\
    0 & \text{if the $i$th ball is thrown into a bad bin,} \\
    -1 & \text{otherwise.}
    \end{cases}
\]

We have $Y \geq \sum_{i=1}^n X_i$. This is because for each $i$, if $X_i = 1$ (the ball is thrown into an empty good bin), $Y$ increases by $1$, and if $X_i = -1$ (the ball is thrown into a possibly non-empty good bin), $Y$ decreases by at most $1$. 

Observe that the $X_i$'s are independent, and $\forall i \in [n]$, we have $\Pr[X_i = 1] \geq (\alpha - \beta - \gamma) / (1 - \gamma)$ and $\Pr[X_i = -1] \leq \beta / (1 - \gamma)$, so we have
\begin{align*}
    \E[X_i] \geq 1 \cdot (\alpha - \beta - \gamma) / (1 - \gamma) + (-1) \cdot \beta / (1 - \gamma) \geq \alpha - 2 \beta - \gamma.
\end{align*}

Using Hoeffding's inequality, together with the fact that $X_i$ is bounded within the interval $[a_i, b_i]$ with $a_i = -1$ and $b_i = 1$,  we have
\begin{align*}
    \Pr\left[Y \leq (\alpha - 2 \beta - \gamma - t) n\right] &\leq \Pr\left[\sum_{i=1}^n X_i \leq (\alpha - 2 \beta - \gamma - t) n\right] \\
    &\leq \Pr\left[\frac{1}{n} \sum_{i=1}^n (X_i - \E[X_i]) \leq -t\right]\\
    &\leq \exp\left(- \frac{2n^2 t^2}{\sum_{i=1}^n (b_i - a_i)^2}\right)\\
    &\leq e^{-\frac{t^2 n}{2}}.\qedhere
\end{align*}
\end{proof}

Using the balls-into-bins procedure of~\cref{lem:good_bins_one_ball} and the ID assignment procedure from the proof of~\cite[Lemma 18]{ChangKPWZ17}, we design the following basic leader election algorithm. 

\begin{lemma}[Basic leader election algorithm]\label{lem:basic_alg}
Let $S$ be a set of devices, where all devices in $S$ agree on an integer $n' \geq 2$.
There is an algorithm in the $\nocd$  model achieving the following goals.
\begin{itemize}
    \item The algorithm costs $O(n')$ time and $O(1)$ energy in the worst case.
    \item The algorithm elects a leader with probability $1 - 2^{-\Omega(n')}$ if $n' / 3 \leq |S| \leq n'$.
\end{itemize}
\end{lemma}
\begin{proof}
We define  $N' = 100n'$.
The algorithm has two phases. 

\paragraph{ID assignment.}
In Phase I, the devices in $S$ contend for the IDs in the ID space $[N']$. Each device $v$ first picks a single \emph{transmitting ID} uniformly at random from $[N']$ and then picks $30$ \emph{listening IDs} uniformly  at random  from the remaining $N'-1$ IDs with duplicates.

For each $i \in [N']$, there are two rounds, and each device has three possible actions:
\begin{enumerate}
    \item If $i$ is a transmitting ID for the device $v$, then $v$ transmits in the first round and listens in the second round.
    \item If $i$ is a listening ID for the device $v$, then $v$ listens in the first round, and $v$ transmits in the second round if $v$ received a message in the first round, otherwise $v$ stays idle in the second round.
    \item If $i$ is neither a transmitting ID nor a listening ID for the device $v$, then $v$ stays idle in both two rounds.
\end{enumerate}

If $i$ is a transmitting ID for a device $v$ and $v$ receives a message in the second round, then the ID $i$ is assigned to this device $v$. It is straightforward to verify that each ID is assigned to at most one device. The algorithm costs $O(n')$ time and $O(1)$ energy.
We write $S' \subseteq S$ to denote the set of the devices that are assigned IDs.

\paragraph{Leader election.}
In Phase II, the devices in $S'$ elect a leader by running the deterministic algorithm of \cite[Theorem 2]{Chang2021detLE} over the space $[N']$. 
The algorithm costs $O(N') = O(n')$ time and $O\left(\log \frac{N'}{|S'|}\right)$ energy. 

For the case $n' / 3 \leq |S| \leq n'$, we will later show that with probability $1 - 2^{-\Omega(n')}$ we have $|S'| = \Theta(N')$, so the leader election algorithm of \cite[Theorem 2]{Chang2021detLE} costs $O(1)$ energy, as required. To ensure that the energy usage of our algorithm never exceeds $O(1)$, we let each device stop participating in the leader algorithm once its energy usage exceeds the required upper bound $O(1)$.

\paragraph{Analysis.} For the rest of the proof, we show that with probability $1 - 2^{-\Omega(n')}$, the number of IDs that are assigned is $\Theta(N')$. Recall that an ID $i \in [N']$ is assigned if it is a listening ID for exactly one device and a transmitting ID for exactly one device.

First, consider the listening IDs. We use \cref{lem:good_bins_one_ball} with $N = N'$, $\alpha=1$, $\beta=30|S|/N'$, $\gamma= 1/N'$, and $t=1/200$. Here we interpret the ID space $[N']$ as $N = N'$ bins, where all of them are good. The $30 |S| = \beta N'$ random choices of listening IDs are seen as $\beta N'$ balls. 
Each ball is thrown to a random bin in a subset of $N' - 1 = (1-\gamma)N$ bins.
We have $0.1 \leq \beta \leq 0.3$ since $n' / 3 \leq |S| \leq n'$ and $N' = 100n'$.
We have $\gamma \leq 1/200$ since $n' \geq 2$ and $N' = 100n'$. 
Hence $\alpha - 2 \beta - \gamma - t \geq 1 - 0.6 - 0.005 - 0.005 > 1/3$, so the probability that the number of IDs in $[N]$ that are assigned as a listening ID to exactly one device is at least   $(\alpha - 2 \beta - \gamma - t)\beta N \geq \beta N /3 \geq N/30$ with probability at least $1 - 2^{-\Omega(n')}$, by  \cref{lem:good_bins_one_ball}.

Next, consider the transmitting IDs. We use \cref{lem:good_bins_one_ball} with $N = N'$, $\alpha=1/30$, $\beta=|S|/N'$, $\gamma=0$, and $t=0.01$.
Again, we interpret the ID space $[N']$ as $N = N'$ bins, but only the ones assigned as a listening ID to exactly one device are considered good. In the following analysis, we condition on the event that the number of IDs assigned as a listening ID to exactly one device is at least $N/30 = \alpha N'$, which occurs with probability $1 - 2^{-\Omega(n')}$. The $|S| = \beta N'$ random choices transmitting IDs are seen as $\beta N'$ balls. 
Each ball is thrown to a bin uniformly at random from the set of all $N' = (1-\gamma)N$ bins. We have $0.01/3 \leq \beta \leq 0.01$ since $n' / 3 \leq |S| \leq n'$ and $N' = 100n'$. 
Hence $\alpha - 2 \beta - \gamma - t \geq 1/30 - 0.02 - 0 - 0.01  = 1/300$, and the probability that the number of IDs in $[N]$ that are assigned as a transmitting ID to exactly one device and assigned as a listening ID to exactly one device   is at least $(\alpha - 2 \beta - \gamma - t)\beta N \geq \beta N /300 \geq N/90000$ with probability at least $1 - 2^{-\Omega(n')}$, as required.
\end{proof}

\paragraph{Main algorithm.}
We extend \cref{lem:basic_alg} to cope with a general failure probability parameter $0 < f < 1$. We show that leader election can be done in worst-case $O(\log f^{-1})$ time and expected $O(\tilde{n}^{-1} \log f^{-1})$ energy, see \cref{lem:basic_alg2} for the precise specification of our algorithm. Later in \cref{sect:derandomize} we will show that both the time and energy complexities of our algorithm are optimal.

Since the expected energy cost  $O(\tilde{n}^{-1} \log f^{-1})$ can be much smaller than one, it is implicit in the statement of \cref{lem:basic_alg2} that 
each device participates in the algorithm with probability $p = \min\{1, O(\tilde{n}^{-1} \log f^{-1}) \}$ independently.
 If a device $v$ chooses to not participate in the algorithm, then $v$ stays idle throughout the algorithm and  its energy usage is zero.
 If a device $v$ chooses to participate in the algorithm, then the energy cost of $v$ is at most $O(1 + \tilde{n}^{-1} \log f^{-1})$.

 In \cref{lem:basic_alg2}, we do not require all devices to be informed whether a leader is elected by the end of the algorithm. Indeed, if $f = 2^{-o(\tilde{n})}$, then the expected energy cost is much less than one, so a majority of the devices do not participate in the algorithm at all, and these devices cannot know the outcome of the algorithm as they remain idle throughout the algorithm.

We can allocate one additional round after running the algorithm of \cref{lem:basic_alg2} to let the elected leader speak to the rest of the devices. This costs one unit of energy for all devices. We choose to not include this step in the algorithm of \cref{lem:basic_alg2} because this will break the bound $O(\tilde{n}^{-1} \log f^{-1})$ of the expected energy cost. Of course, for the case $\tilde{n}^{-1} \log f^{-1} = \Omega(1)$, we can assume that the outcome of the leader election is known to all devices.

\begin{restatable}[Leader election given a network size estimate $\tilde{n}$]{theorem}{thmbasic}
\label{lem:basic_alg2}
Given a number $0 < f < 1$ and an integer $\tilde{n} \geq 2$, there is an algorithm $\mathcal{A}$ achieving the following goals.
\begin{description}
    \item[Expected energy:] The expected energy cost for a device is $O(\tilde{n}^{-1} \log f^{-1})$. 
    \item[Worst-case energy:] The worst-case energy cost for a device is $O(1 + \tilde{n}^{-1} \log f^{-1})$.
    \item[Time:] The time complexity of the algorithm is $O(\log f^{-1})$, which is a fixed number independent of the random bits used by the algorithm.
    \item[Leader election:] By the end of the algorithm, at most one device identifies itself as a leader. If   $\tilde{n}/2 < n \leq \tilde{n}$, then the probability that no leader is elected is at most $f$.
\end{description}
\end{restatable}

\begin{proof}
For each integer $2 \leq n' \leq \tilde{n}$, we define $f_1 = 2^{-\Omega(n')}$ to be the failure probability of leader election in \cref{lem:basic_alg} with parameter $n'$, and we define $f_2 = 2^{-\Omega(n')}$ to be the probability defined as follows. Assuming that the actual network size $n$ satisfies  $\tilde{n}/2 < n \leq \tilde{n}$, pick a subset of devices $S$ by including each device with probability $0.9 \cdot \left(n' / \tilde{n}\right)$ independently. Let $f_2$ be the maximum probability that the inequality $n' / 3 \leq |S| \leq n'$ does not hold, where the maximum ranges over all $n$ satisfying $\tilde{n}/2 < n \leq \tilde{n}$.
By a Chernoff bound, we have $f_2 = 2^{-\Omega(n')}$.

Suppose that there exists a number $2 \leq n' \leq \tilde{n}$ satisfying  $f_1 + f_2 \leq f$, then we simply pick a subset of devices $S$ by including each device with probability $0.9 \cdot\left(n' / \tilde{n}\right)$ and run the algorithm of \cref{lem:basic_alg} with $n'$ and $S$. Such a number $n'$ must satisfies $n' = O(\log f^{-1})$, so each device participates in the algorithm with probability $0.9 \cdot \left(n' / \tilde{n}\right)= O(\tilde{n}^{-1} \log f^{-1})$ independently.
Since the energy cost in \cref{lem:basic_alg}  is $O(1)$, the expected and worst-case energy cost of our algorithm are  $O(\tilde{n}^{-1} \log f^{-1})$ and $O(1)$. By \cref{lem:basic_alg}, the time complexity of our algorithm is $O(n') = O(\log f^{-1})$. A leader is guaranteed to be elected if  $n' / 3 \leq |S| \leq n'$ and the algorithm of \cref{lem:basic_alg} succeeds. This occurs with probability at least $1 - f_1 - f_2 = 1 - f$ whenever $\tilde{n}/2 < n \leq \tilde{n}$.

Suppose that there does not exist a number $2 \leq n' \leq \tilde{n}$ with $f_1 + f_2 \leq f$. Then we must have $f = 2^{-\Omega(\tilde{n})}$ and so $\tilde{n}^{-1} \log f^{-1} = \Omega(1)$. In this case, we simply let $S$ be the set of all devices and run the algorithm of \cref{lem:basic_alg} with $n' = \tilde{n}$ for $C = \Theta(\tilde{n}^{-1} \log f^{-1})$ times. For the case $\tilde{n}/2 < n \leq \tilde{n}$, the probability that no leader is elected in all $C$ iterations is $2^{-\Omega(C\tilde{n})} = 2^{-\Omega(\log f^{-1})}$, which can be made at most $f$ by selecting $C = \Theta(\tilde{n}^{-1} \log f^{-1})$ to be sufficiently large. From the time and energy complexities specified in  \cref{lem:basic_alg}, it is clear the our algorithm satisfies all the requirements in the statement of \cref{lem:basic_alg2}.
\end{proof}

\subsection{Multiple Instances}\label{sect:multi}

Suppose that we are given $k$ pairs $(\tilde{n}_1, c_1), (\tilde{n}_2, c_2), \ldots, (\tilde{n}_k, c_k)$ such that we need to invoke the algorithm of \cref{lem:basic_alg2} with parameters $\tilde{n}_i$ and $f = 2^{-c_i}$ for each $1 \leq i \leq k$. A simple calculation shows that the worst-case energy complexity of the combined algorithm is $O\left(k + \sum_{j=1}^k c_i \tilde{n_i}^{-1} \right)$. 
This bound  can be improved to $O\left(1 + \sum_{j=1}^k c_i \tilde{n_i}^{-1} \right)$ at the cost of allowing the random bits used in different executions of  the algorithm of \cref{lem:basic_alg2} to correlate.

Consider the probability $p_i = \min\{1, O(\tilde{n}_i^{-1} c_i) \}$  that a device participates in an execution of the algorithm of \cref{lem:basic_alg2} with parameters $\tilde{n} = \tilde{n}_i$ and $f = 2^{-c_i}$. We partition the indices $\{1,2, \ldots, k\}$ into groups $[k] = B_1 \cup B_2 \cup \cdots \cup B_s$ such that $1/2 \leq \sum_{i \in B_j} p_i \leq 1$ for each $1 \leq j < s$ and $\sum_{i \in B_j} p_i \leq 1$ for $j=s$. It is clear that such a partition exists, and the  number $s$ of groups is at most $1 + \sum_{j=1}^k 2 p_j = O\left(1 + \sum_{j=1}^k c_j \tilde{n}_j^{-1} \right)$.

Consider any group $B_j$.  Since $\sum_{i \in B_j} p_i \leq 1$, we can ensure that each device only participate in at most one execution of the algorithm of \cref{lem:basic_alg2} associated with the indices in $B_j$. Specifically, we let each device $v$ samples a random variable $x$ such that $x = i$ with probability $p_{i}$ for each $i \in B_j$. Then each device $v$ participates in the execution of the algorithm of \cref{lem:basic_alg2} associated with the index $i \in B_j$ if $x = i$.
Hence the worst-case energy cost for invoking the algorithm of \cref{lem:basic_alg2} for all indices $i \in B_j$ is $\max_{i \in B_j} O(1 + \tilde{n}_i^{-1} c_i) =  O\left( 1 + \sum_{i \in B_j} c_i \tilde{n}_i^{-1} \right)$
 instead of the bound $ O\left( \sum_{i \in B_j} \left(1 + c_i \tilde{n}_i^{-1}\right) \right) = O\left(|B_j| + \sum_{i \in B_j}  c_i \tilde{n}_i^{-1} \right)$ given by the straightforward summation. 

Going over all groups $B_1, B_2, \ldots, B_s$, the overall worst-case energy cost for invoking the algorithm of \cref{lem:basic_alg2} with parameters $(\tilde{n}_1, c_1), (\tilde{n}_2, c_2), \ldots, (\tilde{n}_k, c_k)$  is
\[
 O\left( \sum_{j=1}^s \left(1 + \sum_{i \in B_j} c_i \tilde{n}_i^{-1}\right) \right) 
 = O\left( s +  \sum_{i=1}^k c_i \tilde{n}_i^{-1} \right)
 = O\left(1 + \sum_{i=1}^k c_i \tilde{n}_i^{-1} \right). 
\]

We summarize the discussion as a lemma. This lemma will be used in the algorithms of \cref{sect-rand-main}.

\begin{lemma}[Multiple instances]\label{lem:basic_alg_multi}
Given parameters  $(\tilde{n}_1, c_1), (\tilde{n}_2, c_2), \ldots, (\tilde{n}_k, c_k)$, there is an algorithm $\mathcal{A}$ achieving the following goals in the $\nocd$ model.
\begin{description}
    \item[Energy:] The worst-case energy cost for a device is $O\left( 1 + \sum_{j=1}^k c_j \tilde{n}_j^{-1} \right)$.
    \item[Time:] The time complexity of the algorithm is $O\left( \sum_{j=1}^k c_j \right)$, which is a fixed number independent of the random bits used by the algorithm.
    \item[Leader election:] For each $1 \leq j \leq k$, at most one device identifies itself as the $j$th leader. The probability that the $j$th leader is not elected is at most $2^{-c_j}$ if $\tilde{n}_j/2 < n \leq \tilde{n}_j$.
\end{description}
\end{lemma}

\subsection{Derandomization}\label{sect:derandomize}

By derandomizing \cref{lem:basic_alg2}, we obtain an optimal deterministic leader election algorithm for the $\nocd$ model with  $O(n \log \frac{N}{n})$ time and $O(\log \frac{N}{n})$ energy, where the size  $N$ of the ID space $[N]$  and an estimate $\tilde{n}$ of the  number of devices $n$ such that  $\tilde{n}/2 < n \leq \tilde{n}$ are both known to all devices.
As we will later see, the optimality of the deterministic algorithm implies the optimality of \cref{lem:basic_alg2}.

\thmoptdet*

\begin{proof}
In  \cref{lem:basic_alg2} we pick \[f = \frac{1}{1 + \sum_{\tilde{n}/2 < n \leq \tilde{n}} \binom{N}{n}}.\]
Consider an assignment $\phi$ from $[N]$ to an infinite sequence of random bits.
When we run the randomized algorithm of \cref{lem:basic_alg2} on a device $v$ with identifier $i$, $v$ uses the random bits $\phi(i)$.
For each fixed size-$n$ subset of $[N]$ with $\tilde{n}/2 < n \leq \tilde{n}$, the simulation of the randomized algorithm of \cref{lem:basic_alg2} successfully elects a leader  with probability at least $1 - f$.
By a union bound, there is a non-zero probability that the simulation of the randomized algorithm of \cref{lem:basic_alg2} succeeds for all size-$n$ subsets of $[N]$ such that  $\tilde{n}/2 < n \leq \tilde{n}$. This non-zero probability implies the existence of a deterministic algorithm that works for all size-$n$ subsets of $[N]$ such that  $\tilde{n}/2 < n \leq \tilde{n}$. Since $\tilde{n}/2 < n \leq \tilde{n}$ and $\log f^{-1} =  O\left(n \log \frac{N}{n}\right)$, this  algorithm has time complexity $T = O(\log f^{-1}) = O\left(n \log \frac{N}{n}\right)$ and energy complexity $E  = O(\tilde{n}^{-1} \log f^{-1}) = O\left(\log \frac{N}{n}\right)$, by \cref{lem:basic_alg2}.
\end{proof}

\paragraph{Optimality of \cref{thm:optimal_algo_no_cd}.} The algorithm of \cref{thm:optimal_algo_no_cd} is energy-optimal in $\nocd$ due to a   matching energy lower bound of $\Omega\left(\log \frac{N}{n}\right)$ in~\cite[Theorem 3]{Chang2021detLE}, which even applies to the case where the number $n$ of devices is  global knowledge and $n$ can be any integer in the range $[2, N-1]$.

The algorithm of \cref{thm:optimal_algo_no_cd} is  also time-optimal in $\nocd$. It is implicit in the proof of~\cite[Theorem 3.3]{clementi2003distributed} that there is a time lower bound of  $\Omega\left(n \log \frac{N}{n}\right)$ which applies to the case where the number $n$ of devices is  global knowledge and $n$ can be any integer in the range $[2, N-1]$ that is a power of $2$.

Specifically, a $t$-time deterministic leader election algorithm that works for $n$ and $N$ in $\nocd$ can be seen as a family $\mathcal{F}$ of $t$ subsets of $[N]$ such that for each size-$n$ subset $S$ of $N$ there is a subset $F \in \mathcal{F}$ with $|F \cap S| = 1$.

For the case that $2 \leq n \leq N/64$ and $n$ is a power of $2$, 
the proof of~\cite[Theorem 3.3]{clementi2003distributed} shows that such a family $\mathcal{F}$ must have size $t = \Omega \left(n\log \frac{N}{n}\right)$. Here the parameter $k$ in the proof of~\cite[Theorem 3.3]{clementi2003distributed} is our $n$, and the parameter $n$  in the proof of~\cite[Theorem 3.3]{clementi2003distributed} is our $N$. Since $k$ is a power of $2$, we have $k'=k$  in the proof of~\cite[Theorem 3.3]{clementi2003distributed}. Hence a time lower bound of  $\Omega\left(n \log \frac{N}{n}\right)$ is obtained.

For the remaining case that $N/64 < n \leq N-1$, we can apply the $\Omega(N) = \Omega\left(n \log \frac{N}{n}\right)$ time lower bound of~\cite[Theorem 1.6]{JurdzinskiKZ02podc}.

\paragraph{Optimality of \cref{lem:basic_alg2}.} 
A corollary of above discussion is that the time $O(\log f^{-1})$ and energy $O(\tilde{n}^{-1} \log f^{-1})$ complexities of \cref{lem:basic_alg2} are optimal. If  the time complexity of \cref{thm:optimal_algo_no_cd} can be made $o(\log f^{-1})$, then we can derandomize it to give a deterministic leader election algorithm whose time complexity is $o\left(n \log \frac{N}{n}\right)$, violating the $\Omega\left(n \log \frac{N}{n}\right)$ lower bound.
Similarly, if the energy complexity of \cref{thm:optimal_algo_no_cd} can be made $o(\tilde{n}^{-1}  \log f^{-1})$, then   we can derandomize it to give a deterministic leader election algorithm whose energy complexity is $o\left(\log \frac{N}{n}\right)$, violating the $\Omega\left(\log \frac{N}{n}\right)$ lower bound.

\section{Leader Election with an Unknown Number of Devices}\label{sect-rand-main}

In this section, we design randomized leader election algorithms in the model  in the scenario where the number of devices $n$ is completely unknown.  In \cref{sect-framework} we describe the  framework for our leader election algorithms. To illustrate the use of our framework. in \cref{sect-nocd-UB}, we reprove a result of~\cite{lavault2007quasi} using our framework
that in $\nocd$ a leader can be elected in expected $O(\log n)$ time and using expected $O(\log \log n)$ energy.

\subsection{Basic Framework}\label{sect-framework}

Our randomized leader election algorithms proceeds in iterations. We will specify an infinite sequence of deadlines $(d_1, d_2, \ldots)$ such that the $i$th iteration finishes by time  $O(d_i)$. Alternatively, we require the algorithm of the $i$th iteration to take at most $O(d_i - d_{i-1})$ time, and we let $d_0 = 0$ for convenience.

\paragraph{One iteration.}
Each iteration of the algorithm is specified by a list of pairs of positive integers \[(\tilde{n}_1, c_1), (\tilde{n}_2, c_2), \ldots, (\tilde{n}_k, c_k).\] Given such a list, the algorithm of one iteration has the following three parts. 
\begin{description}
    \item[Using the basic subroutine:] In the first part, we run the algorithm of \cref{lem:basic_alg_multi} with parameters $(\tilde{n}_1, c_1), (\tilde{n}_2, c_2), \ldots, (\tilde{n}_k, c_k)$. 
    \item[Leader election:] In the second part, 
we allocate $O(k)$ rounds to elect one leader among the at most $k$ leaders elected during the execution of the algorithm of \cref{lem:basic_alg_multi}.  More specifically, we allocate an ID space $[N]$ with $N=k$. For each $1 \leq j \leq k$, if a device $v$ is the $j$th leader elected in  \cref{lem:basic_alg_multi}, then $v$  grabs the ID $j$. Then we run the well-known $O(N)$-time and $O(\log N)$-energy deterministic $\nocd$ algorithm~\cite{ChangKPWZ17} to elect a single leader among these devices. It is possible that a device $v$ grabs multiple IDs. In this case, the  ID of $v$ in the deterministic algorithm is set to be the smallest identifier among all identifiers that $v$ grabs.
\item[Announcing the result:] In the third part, we  allocate one round to let the leader elected during the second part of the algorithm speak, while all other devices listen. 
If a leader is elected, the algorithm terminates. 
Otherwise, the algorithm moves on to the next iteration.
\end{description}

It is possible that no leader is elected, as there is no guarantee that the algorithm of \cref{lem:basic_alg_multi} must elect a leader.

\paragraph{Time.} The time complexity of the algorithm is dominated by the round complexity of the first part, which is $O\left(\sum_{j=1}^k c_j \right)$ according to \cref{lem:basic_alg_multi}. Here we assume that each $c_j$ is a positive integer, so $\sum_{j=1}^k c_j \geq k$.  The number of rounds used by the algorithm is fixed independent of the randomness, so the time complexity upper bound holds in the worst case.

\paragraph{Energy.}   For the first part of the algorithm, the energy complexity is $O\left(1 + \sum_{j=1}^k c_i \tilde{n}_i^{-1} \right)$ by \cref{lem:basic_alg_multi}.
For the second part of the algorithm, the energy cost is $O(\log k)$ among those devices participating in the leader election algorithm, and it is zero for the rest of the devices. For the third part of the algorithm, the energy cost is one unit for all devices. All these upper bounds hold in the worst case.

\subsection{An \texorpdfstring{$O(\log n)$}{O(log n)}-time and \texorpdfstring{$O(\log \log n)$}{O(log log n)}-energy Algorithm in \texorpdfstring{$\nocd$}{No-CD}}\label{sect-nocd-UB}

Using the framework of algorithm design in \cref{sect-framework}, we reprove a result of~\cite{lavault2007quasi}
that in the $\nocd$ model a leader among an unknown number $n \geq 2$ of devices can be elected using $O(\log n)$ time and $O(\log \log n)$ energy in expectation.  Our algorithm is described by the following choice of parameters.

\begin{description}
    \item[The deadline sequence:] We set $d_i = 2^i$ for each $i \geq 1$.
    \item[The parameters for each iteration:] The algorithm of the $i$th iteration is specified by the list of pairs $(\tilde{n}, 2)$ for all $\tilde{n} = 2^{1}, 2^2, \ldots, 2^{d_i - 1}$.
\end{description}

With the above choice  of parameters, the algorithm for the $i$th iteration takes $O\left(\sum_{j=1}^{d_i-1} 2\right) = O(d_i) = O(d_i - d_{i-1})$ time, so the choice of the deadline sequence is valid.

\paragraph{Failure probability.} Suppose that we run the algorithm on a network of $n \geq 2$ devices. We pick $\tilde{n}=2^{j^\ast}$ to be the number such that  $\tilde{n}/2 < n \leq \tilde{n}$ and $j^\ast$ is an integer. Let $i^\ast$ be the smallest index $i$ such that $\log \tilde{n} < d_i$, so the pair $(\tilde{n}, 2)$ is included in the specification for each iteration $i \geq i^\ast$. Therefore, for each $i \geq i^\ast$, if we run the $i$th iteration of the algorithm on a network of $n \geq 2$ devices, then the probability that no leader is elected in this iteration is at most $1/4$, as $\tilde{n}/2 < n \leq \tilde{n}$ and  the pair $(\tilde{n}, 2)$ is considered in the $i$th iteration.
We define
\[
f_i = 
\begin{cases}
1 & i < i^\ast,\\
2^{-2(i - i^\ast + 1)} & i \geq i^\ast,
\end{cases}
\]
so  
$f_i$ is an upper bound on the probability that no leader is elected by the end of the $i$th iteration. In other words, the probability that the algorithm enters the $(i+1)$th iteration is at most $f_i$.


\paragraph{Expected time complexity.}  Since the time spent on the $(i+1)$th iteration is $O(d_{i+1} - d_{i})$, the expected time complexity $T(n)$ of the algorithm can be upper bounded as follows.
\begin{align*}
T(n)
& = O(d_{i^\ast}) + \sum_{i = i^\ast}^\infty f_i \cdot O(d_{i+1} - d_{i})\\
& = O(2^{i^\ast}) + \sum_{i = i^\ast}^\infty 2^{-2(i - i^\ast + 1)} \cdot O(d_{i})\\
& = O(2^{i^\ast}) + \sum_{i = i^\ast}^\infty O\left(2^{i-2(i - i^\ast + 1)}\right)\\
& = O(2^{i^\ast}) \cdot O\left(1 + \sum_{i = i^\ast}^\infty 2^{(i - i^\ast)-2(i - i^\ast + 1)}\right)\\
& = O(2^{i^\ast}) \cdot O\left(1 + \sum_{j = 0}^\infty 2^{-2-j}\right)\\
&= O(2^{i^\ast}) \\
&= O(\log n).
\end{align*}
Here we use the property that $d_{i^\ast} = 2^{i^\ast} = O(\log n)$ by our choice of $i^\ast$.




\paragraph{Expected energy complexity.} 
The energy cost for the algorithm of one iteration depends on whether a leader is elected in this iteration. Given the parameters $(\tilde{n}_1, c_1), (\tilde{n}_2, c_2), \ldots, (\tilde{n}_k, c_k)$,  the energy cost is $O\left(1 + \sum_{j=1}^k c_j \tilde{n}_j^{-1} \right)$ if no leader is elected, and it is $O\left(1 + \sum_{j=1}^k c_j \tilde{n}_j^{-1} \right) + O(\log k)$ if a leader is elected. These bounds hold in the worst case.

With our parameters, the energy cost for the  $i$th iteration is $O\left(1 + \sum_{j=1}^{d_i-1} 2 \cdot 2^{-j} \right) = O(1)$ if no leader is elected in this iteration, and it is $O\left(1 + \sum_{j=1}^{d_i-1} 2 \cdot 2^{-j} \right) + O(\log d_i) = O(i)$ if a leader is elected in this iteration. Once a leader is elected, the algorithm is terminated, so
 the energy cost of the entire algorithm is $O(i)$ if the algorithm is terminated by the end of the $i$th iteration as a leader is elected in the $i$th iteration. Therefore, the expected energy cost $E(n)$ of the algorithm can be upper bounded by $O\left( \sum_{i=0}^\infty f_i \right)$, and we have
\[E(n) = O\left( \sum_{i=0}^\infty f_i \right) = O(i^\ast) +  O\left(\sum_{j=0}^\infty 2^{-2-2j}\right) = O(i^\ast) = O(\log \log n).\]

We summarize the discussion as a theorem.

\begin{restatable}[{Leader election with an  unknown number of devices~\cite{lavault2007quasi}}]{theorem}{thmnocd}
\label{thm:algo_nocd}
There is an algorithm in the $\nocd$ model that elects a leader from an unknown number  $n \geq 2$ of devices using expected $O(\log n)$ time and expected $O(\log \log n)$ energy.
\end{restatable}

\section{An exponential improvement with sender collision detection} \label{sect:sender-main}

In this section, we apply the framework introduced in \cref{sect-rand-main} to prove \cref{thm:algo_sendercd}.
In \cref{sect:primitive-sender}, we describe a different leader election subroutine that is exponentially more energy-efficient in $\sendercd$. 
In \cref{sect-sendercd-UB} we show that in $\sendercd$ a leader can be elected in $O(\log^{1+\epsilon} n)$ time and using $O(\epsilon^{-1}\log \log \log n)$ energy in expectation, where $\epsilon > 0$ can be an arbitrarily small constant.

\subsection{A Subroutine for \texorpdfstring{$\sendercd$}{Sender-CD}} \label{sect:primitive-sender}

In this section, we prove the following lemma, which will be used in our $\sendercd$ algorithm in \cref{sect-sendercd-UB} to deal with small $\tilde{n}$-values.

\begin{lemma}[A subroutine for $\sendercd$]\label{lem:primitive-sender}
Let $d$ be a parameter known by all devices.
There is an  algorithm achieving the following goals in the $\sendercd$ model.
\begin{description}
    \item[Leader election:] The algorithm elects at most one leader. If the number $n$ of devices satisfies $n =  O(\log^2 d)$, then with probability at least $1 - d^{-3}$ a leader is elected.
    \item[Time:] The time complexity of the algorithm is $O(d)$, which is a fixed number independent of the random bits used by the algorithm. 
    \item[Energy:] If no leader is elected, then the energy cost is $O(1)$. If a leader is elected, then the energy cost is $O(\log \log d)$.
\end{description}
\end{lemma}
\begin{proof}
We show how the deterministic $\sendercd$ algorithm  in~\cite[Section 2.1]{Chang2021detLE} can be applied to achieve this goal.
We first review this deterministic  algorithm. 

\paragraph{Good partitions.}
Let $N$ be the size of the ID space $[N]$.
There are two parameters known by all device:   $b$ is a positive integer and $0 < \tilde{\epsilon} < 1$ is a constant.
Given the two parameters, it was shown in~\cite{Chang2021detLE} that there is a family of $K = O(\tilde{\epsilon}^{-1} \log_b N)$ partitions of $[N]$ into $b$ parts, $[N]=S_{1}^{(i)} \cup S_{2}^{(i)} \cup \cdots \cup S_{b}^{(i)}$ for each $i \in [K]$, meeting the following condition.
\begin{itemize}
    \item For each  $V \subseteq [N]$ with $1 \leq |V| \leq b^{1 - \tilde{\epsilon}}$, there exist $i \in [K]$ and $j \in [b]$ such that $\left|S_{j}^{(i)} \cap V\right| = 1$.
\end{itemize}

\paragraph{Deterministic algorithm.}
Traditionally in the deterministic setting we assume that each ID $x \in [N]$ is assigned to at most one device, but in our setting we have to allow the possibility that an  ID $x \in [N]$ is assigned to more than one device.

Using this family of partitions, the following algorithm elects a leader when the number of devices  assigned a unique identifier is at most $b^{1 - \tilde{\epsilon}}$. The algorithm has $K = O(\tilde{\epsilon}^{-1} \log_b N)$ iterations. Each iteration $i$ consists of two parts. 
\begin{description}
    \item[First part:] The algorithm of the first part uses $b$ rounds. For $1 \leq j \leq b$, in the $j$th round we let each device $v$ with $\ID(v) \in S_{j}^{(i)}$ transmits simultaneously, so $v$ learns whether $v$ is the only device with $\ID(v) \in S_{j}^{(i)}$ according to whether it  hears back its message. This is possible as we are in the $\sendercd$ model.
    \item[Second part:] The algorithm of the second part runs the 
     $O(N')$-time and $O(\log \log N')$-energy deterministic $\sendercd$ leader election algorithm of~\cite{ChangKPWZ17} with $N' = b$.
    For each $j \in [b]$, if there is exactly one device $v$ with $\ID(v) \in S_{j}^{(i)}$, then $v$ participates in the leader election algorithm using the new ID $j$. 
\end{description}
If a leader is successfully elected in the second part, then the algorithm terminates, otherwise the algorithm proceeds to the next iteration. We analyze this deterministic algorithm.

\begin{description}
    \item[Time:] In each iteration, both the first and the second parts costs $O(b)$ time, so the overall time complexity is $O(Kb) = O(b \tilde{\epsilon}^{-1} \log_b N)$.
    \item[Energy:] In each iteration, if a device does not participate in the leader election algorithm in the second part, then its energy cost is $O(1)$, otherwise its energy cost is $O(\log \log b)$. Since the algorithm terminates once a leader is elected, the overall energy cost is $O(K) =  O(\tilde{\epsilon}^{-1} \log_b N)$ if no leader is elected, otherwise the overall energy cost is $O(K + \log \log b) = O(\tilde{\epsilon}^{-1} \log_b N + \log \log b)$.
    \item[Correctness:] It is clear that at most one leader is elected by the algorithm. If the number of devices  assigned a unique identifier is at most $b^{1 - \tilde{\epsilon}}$, then it is guaranteed that a leader is elected, as there exist $i \in [K]$ and $j \in [b]$ such that there is exactly one device $v$ with $\ID(v) \in S_{j}^{(i)}$, so there will be at least one device participating in the leader election algorithm of the second part of iteration $i$.
\end{description}

\paragraph{Randomized algorithm.}
The idea of our randomized algorithm is simply that we choose $N = d^{4}$ and let each device select an ID from the ID space $[N]$ uniformly at random to run the above deterministic algorithm with $b = d$ and $\tilde{\epsilon}=1/2$, so the time complexity is $O(b \tilde{\epsilon}^{-1} \log_b N) = O(d)$, and the energy complexity is $O(\tilde{\epsilon}^{-1} \log_b N) = O(1)$ if no leader is elected, otherwise the overall energy cost is $O(\tilde{\epsilon}^{-1} \log_b N + \log \log b) = O(\log \log d)$.

The algorithm elects at most one leader. Suppose the number $n$ of devices is already at most $n =  O(\log^2 d) \ll \sqrt{d} = b^{1 - \tilde{\epsilon}}$ Then the only reason that no leader is elected is that there are repeated identifiers. The probability that there exist two devices choosing the same identifier is at most $n^2/ N = O(d^{-4} \log^2 d) \ll d^{-3}$, as required.
\end{proof}

\subsection{An \texorpdfstring{$O(\log^{1+\epsilon} n)$}{O(log n to the power of 1+eps)}-time and \texorpdfstring{$O(\epsilon^{-1} \log \log \log n)$}{O(log log log n)}-energy Algorithm in \texorpdfstring{$\sendercd$}{Sender-CD}} \label{sect-sendercd-UB}

We design a leader election algorithm that takes $O(\log^{1+\epsilon} n)$ time and $O(\epsilon^{-1} \log \log \log n)$ energy in expectation in $\sendercd$, for any $0 < \epsilon \leq O(1)$. 
The algorithm incorporates \cref{lem:primitive-sender} into our framework described in \cref{sect-framework}. Intuitively, \cref{lem:primitive-sender} will be used in place of \cref{lem:basic_alg_multi} to deal with network size estimates $\tilde{n}$ that are sufficiently small.
We consider the following parameters.
\begin{align*}
    d_i &= 
    \begin{cases}
2^{\left\lceil(1+\epsilon)^i\right\rceil}, & i \geq 1,\\
0, & i=0.\\
\end{cases}
\\
n_j &= 2^j.\\
m_{i,j}^\ast &=  \max\left\{0, 2\left\lceil\log d_i - \log \log n_j\right\rceil\right\}.\\
m_{i,j} &= \min\left\{\left\lceil 2^{j/2}\right\rceil, m_{i,j}^\ast \right\}.\\
c_{i,j} &= m_{i,j} - m_{i-1, j}.
\end{align*}

\begin{description}
    \item[The deadline sequence:] The deadline sequence of our algorithm is set to be $(d_1, d_2, \ldots)$. This choice of the deadlines $i = \Theta(\epsilon^{-1} \log \log d_i)$ is exponentially sparser than the one $i = \Theta(\log d_i)$ in the $\nocd$ algorithm of \cref{thm:algo_nocd}. 
    This change is needed because here we aim for an exponential improvement in the energy complexity over the algorithm of \cref{thm:algo_nocd}. In our framework, all devices have to spend at least one unit of energy in each iteration until a leader is elected. Since we know that leader election needs $\Omega(\log n)$ time to solve in expectation~\cite{Newport14}, it is necessary that the deadline sequence satisfies $i = O(\epsilon^{-1} \log \log d_i)$ in order to achieve an expected energy complexity of $O(\epsilon^{-1} \log \log \log n)$.
    \item[The parameters for each iteration:] For each iteration $i$ and each positive integer $j$,  the pair $(\tilde{n}_j, c_{i,j})$ is added to the specification of the algorithm of the $i$th iteration if $c_{i,j} \neq 0$. Hence \[2^{-m_{i,j}} = \prod_{l=1}^i 2^{-c_{l,j}}\] is an upper bound on the probability that no leader is elected by the end of iteration $i$ when  $\tilde{n}_j/2 < n \leq  \tilde{n}_j$. 
    We briefly compare the choice of parameters with the parameters in our $\nocd$ algorithm of \cref{thm:algo_nocd}. 
    In our $\nocd$ algorithm, for each iteration $i$ and each network size estimate $n_j = 2^j$, $(n_j, 2)$ is included in the list of pairs if $j < d_i$. Since $d_i = 2^i$ in our $\nocd$ algorithm, 
     the probability that no leader is elected by the end of iteration $i$ when  $\tilde{n}_j/2 < n \leq  \tilde{n}_j$ is at most
    \[2^{-m_{i,j}^\ast} = 2^{-\max\left\{0, 2\left\lceil\log d_i - \log \log n_j\right\rceil\right\}}.\]
 Therefore, in a sense, the only change we introduce in the $\sendercd$ model is to make $m_{i,j} = \min\left\{\left\lceil 2^{j/2}\right\rceil, m_{i,j}^\ast \right\}$ upper bounded by 
 $\left\lceil 2^{j/2}\right\rceil$. 
 Intuitively, this means that once the failure probability associated with a network size estimates $\tilde{n}_j$ becomes sufficiently small $2^{-{\left\lceil 2^{j/2}\right\rceil}}$, we stop further improving the failure probability using \cref{lem:basic_alg_multi}.
 As we will later see, for those network size estimates $\tilde{n}_j$, we will switch to a different approach based on \cref{lem:primitive-sender} to reduce the failure probability to the desired bound $m_{i,j}^\ast$.
 Observe that these network size estimates $\tilde{n}_j$ are necessarily small.   Indeed, $m_{i,j} = \left\lceil 2^{j/2}\right\rceil$ implies that $\tilde{n}_j = O(\log^2 d_i)$, as  $2^{j/2} = \sqrt{\tilde{n}_j}$.
\end{description}

We verify that the time complexity of each iteration $i$ is $O(d_i) = O(d_i - d_{i-1})$, so our choice of the deadline sequence is valid. According to the time complexity analysis in \cref{sect-framework}, the time complexity of iteration $i$ is linear in \[\sum_{j=1}^{\infty} c_{i,j} 
\leq \sum_{j=1}^{\infty}  m_{i,j} 
\leq  \sum_{j=1}^{\infty} m_{i,j}^\ast 
= \sum_{j=1}^{d_i - 1} m_{i,j}^\ast =  \sum_{j=1}^{d_i - 1} \left\lceil\log d_i - \log j \right\rceil  = O(d_i),\] 
where we use the fact that $m_{i,j}^\ast = 0$ for each $j \geq d_i$ and $m_{i,j}^\ast = \left\lceil\log d_i - \log j \right\rceil$ for each $1 \leq j < d_i$.

\paragraph{Modifications.} We  make the following  modifications to the framework described in \cref{sect-framework}.
\begin{itemize}
    \item For the second part of the algorithm of one iteration, we use the $O(N)$-time and $O(\log \log N)$-energy $\sendercd$ algorithm of~\cite{ChangKPWZ17}, which is exponentially more efficient than the $\nocd$ algorithm, so the term $O(\log k)$ in the energy complexity analysis in \cref{sect-framework} can be replaced by $O(\log \log k)$.
    \item As discussed earlier, in the algorithm for iteration $i$, we will invoke the algorithm of \cref{lem:primitive-sender} with $d = d_i$. The objective is to ensure that by the end of   iteration $i$, for each network size estimate $\tilde{n}_j = 2^j$,  if $\tilde{n}_j/2 < n \leq   \tilde{n}_j$, then the probability that no leader is elected is at most $2^{-m_{i,j}^\ast}$ by the end of iteration $i$. If $\tilde{n} = \Omega(\log^2 d_i)$ is sufficiently large, then we already have $m_{i,j}^\ast = m_{i,j}$. As we know that $2^{-m_{i,j}}$ is an upper bound on the probability that no leader is elected by the end of iteration $i$ if $\tilde{n}_j/2 < n \leq   \tilde{n}_j$, we only need to focus on the case $n = O(\log^2 d_i)$. Since $2^{-m_{i,j}^\ast} \geq  2^{-2 \ceil{\log d_i}} = (1/2) \cdot d_i^{-2} \geq d_i^{-3}$, it suffices to design an algorithm that elects a leader with probability at least $1 - d_i^{-3}$ when $n = O(\log^2 d_i)$.
    Hence we may use the algorithm of \cref{sect:primitive-sender} with $d = d_i$ to achieve this goal. The algorithm costs $O(d_i) = O(d_i - d_{i-1})$ time, which is within the time constraint allocated for the $i$th iteration of the algorithm. If no leader is elected, then the energy cost is $O(1)$, otherwise the energy cost is $O(\log \log d_i)$. 
\end{itemize}


\paragraph{Analysis.}
Suppose  we run the algorithm on a network of $n \geq 2$ devices. 
 For the rest of the section, we analyze the expected time $T(n)$ and energy $E(n)$ complexities of our algorithm. 
 Similar to the analysis in \cref{sect-nocd-UB}, we pick $\tilde{n}=2^{j^\ast}$ to be the number such that $\tilde{n}/2 < n \leq \tilde{n}$ and $j^\ast$ is an integer, and
we pick $i^\ast$ to be the smallest index $i$ such that $\log \tilde{n} < d_i$. It is clear that $d_{i^\ast} = O(\log^{1+\epsilon}n)$ and $i^\ast = O(\epsilon^{-1} \log \log \log n)$.
Due to the use of the algorithm of \cref{sect:primitive-sender} in the above modification, 
\[ 
f_i = 2^{-m_{i,j^\ast}^\ast} =
\begin{cases}
1, & i < i^\ast,\\
2^{- 2\left\lceil\log d_i - \log j^\ast\right\rceil} & i \geq i^\ast
\end{cases}
\] 
is an upper bound on the probability that no leader is elected by the end of the $i$th iteration.

\paragraph{Expected time complexity.} 
The expected time complexity of the algorithm can be upper bounded as follows.
\begin{align*}
    T(n) &= \sum_{i=1}^\infty f_i \cdot O(d_i - d_{i-1})\\
    &= O(d_{i^\ast}) + \sum_{i = i^\ast + 1}^\infty f_i \cdot O(d_i - d_{i-1})\\
    &=O(d_{i^\ast}) + \sum_{i = i^\ast + 1}^\infty O\left(2^{- 2\left(\log d_i - \log j^\ast\right)} \cdot d_i\right)
    & f_i = 2^{- 2\left\lceil\log d_i - \log j^\ast\right\rceil} \  \text{ for } \  i \geq i^\ast
    \\
    &=O(d_{i^\ast}) + \sum_{i = i^\ast + 1}^\infty O\left(2^{- 2\left(\log d_i - \log d_{i^\ast}\right)} \cdot d_i\right)
    & d_{i^\ast} > \log \tilde{n} = j^\ast\\
    &=O(d_{i^\ast}) \cdot O\left(1 + \sum_{i = i^\ast + 1}^\infty  2^{-\left(\log d_i - \log d_{i^\ast}\right)} \right)\\
    &=O(d_{i^\ast}) \cdot O\left(1 + \sum_{i = i^\ast + 1}^\infty  2^{-(1+\epsilon)^i+ (1+\epsilon)^{i^\ast}} \right)
    & d_i = 2^{\left\lceil(1+\epsilon)^i\right\rceil}\\
    &=O(d_{i^\ast})\\
    &=O\left(\log^{1+\epsilon}n\right).
\end{align*}




\paragraph{Expected energy complexity.}
The energy cost for the algorithm of one iteration depends on whether a leader is elected in this iteration. For iteration $i$, the list of pairs are $(\tilde{n}_j, c_{i,j})$ over all $1 \leq j < d_i$ such that $c_{i,j} \neq 0$. Hence the energy cost is $O\left(1 + \sum_{j=1}^{d_i - 1} c_{i,j} \tilde{n}_j^{-1} \right)$ if no leader is elected, and it is $O\left(1 + \sum_{j=1}^{d_i - 1} c_{i,j} \tilde{n}_j^{-1} \right) + O(\log \log d_i)$ if a leader is elected. 
We first show that the summation of the  energy cost due to the term $\sum_{j=1}^{d_i - 1} c_{i,j} \tilde{n}_j^{-1}$ is $O(1)$, so it can be ignored in the subsequent analysis. The summation of this term over all iterations $i$ is at most \[\sum_{i=1}^\infty \sum_{j=1}^\infty c_{i,j} \tilde{n}_j^{-1} = \sum_{j=1}^\infty \left\lceil 2^{j/2} \right\rceil  \cdot \tilde{n}_j^{-1} = O(1),\] as  $\sum_{i=1}^{\infty} c_{i,j} 
= \left\lceil 2^{j/2} \right\rceil$ and $\tilde{n}_j = 2^j$.
Therefore, if a leader is elected in iteration $\bar{i}$, then the energy cost of the algorithm is at most
\[O(1) + \sum_{i = 1}^{\bar{i}} O(1) + O(\log \log d_{\bar{i}}) = O(\bar{i}).\] 
In the calculation, the first term $O(1)$ captures the energy cost due to the term $\sum_{j=1}^{d_i - 1} c_{i,j} \tilde{n}_j^{-1}$ over all $i$, the second term $\sum_{i = 1}^{\bar{i}} O(1) = O(\bar{i})$ is the dominant term because $i =  \Theta(\epsilon^{-1} \log \log d_i)$.

Since the expected energy complexity $E(n)$ is linear in the number of iterations until a leader is elected, $E(n)$ can be upper bounded as follows.
\begin{align*}
    E(n) &= O\left(\sum_{i=1}^\infty f_i\right)\\
    &= O(i^\ast) + O\left(\sum_{i = i^\ast + 1}^\infty 2^{- 2\left(\log d_i - \log j^\ast\right)}\right)\\
    &= O(i^\ast) + O\left(\sum_{i = i^\ast + 1}^\infty 2^{- 2\left(\log d_i - \log d_{i^\ast}\right)}\right)\\
    &= O(i^\ast) + O\left(\sum_{i = i^\ast + 1}^\infty 2^{-2(1+\epsilon)^i+ 2(1+\epsilon)^{i^\ast}}\right)\\
    &= O(i^\ast)\\
    &= O(\epsilon^{-1} \log \log \log n).
\end{align*}

We conclude the proof of \cref{thm:algo_sendercd}.

\thmcd*

\addcontentsline{toc}{section}{References}
\bibliographystyle{alpha}
\bibliography{ref}

\appendix

\section{Proof of \texorpdfstring{\cref{lem:logn_time_lb}}{Lemma \ref{lem:logn_time_lb}}}\label{app:missing-proof}

In this section we provide a full proof of \cref{lem:logn_time_lb}.
In the proof, we consider the following scenario. Let $N \geq 2$ be a positive integer. We are given $N$ devices, each with a unique ID in $[N]$. We assume that the number $N$ is known to all devices. At the beginning of an algorithm $\mathcal{A}$, a subset $H$ of the $N$ devices are activated. The subset $H$ is unknown, except that it is guaranteed to have size $\sqrt{N} \leq |H| \leq N$.

As we do not consider the energy complexity, we may assume that each device either transmits or listens in each round. 
We assume that once a collision-free transmission occurs, the algorithm terminates. Indeed, collision-free transmission occurs in a round $t$ if there is exactly one device $v$ transmitting in the round $t$, so $v$ may identify itself as a leader, and all other devices may identify themselves as non-leaders.

In $\sendercd$ and $\nocd$, the feedback from the communication channel is always silence if no successful transmission occurs.
As discussed earlier, the action of a device $v$ in a round depends only on its unique identifier $\ID(v)$ and its private random bits. 
Therefore, a $t$-time deterministic algorithm can be described by a function $\phi$ that maps each $x \in [N]$ to an infinite sequence of actions $(a_1(x), a_2(x), \ldots, a_t(x))$, where each $a_i(x) \in \{\transmit, \listen\}$ indicates the action of a device $v$ with $\ID(v) = x$ in round $i$ given that no successful transmission occurs in the first $i-1$ rounds. If the algorithm is randomized, then the mapping is randomized. 

\paragraph{Deterministic algorithms.} Given a parameter $0 < f < 1$ and a multiset $\mathcal{H}$ of subsets $H \subseteq N$ with $\sqrt{N} \leq |H| \leq N$, we say that a  deterministic algorithm $\mathcal{A}$ succeeds with probability $1-f$ w.r.t.~$\mathcal{H}$  if a collision-free transmission occurs  with probability at least $1-f$ when we run $\mathcal{A}$ on a random subset of devices $H \in \mathcal{H}$.
We write $T^\deterministic_{\mathcal{H}, f}$ to denote the minimum number $t$ such that there exists a $t$-time deterministic algorithm $\mathcal{A}$  that  succeeds with probability $1-f$ w.r.t.~$\mathcal{H}$.

\paragraph{Randomized algorithm.}
Given a parameter $0 < f < 1$ and a subset $H \subseteq N$ with $\sqrt{N} \leq |H| \leq N$, we say that a  randomized algorithm $\mathcal{A}$ succeeds with probability $1-f$ w.r.t.~$H$ if a collision-free transmission occurs  with probability at least  $1-f$ when we run $\mathcal{A}$ on the subset of devices $H \in \mathcal{H}$.

We write $T^\randomized_{f}$ to denote the minimum number $t$ such that there exists a $t$-time randomized algorithm $\mathcal{A}$ such that $\mathcal{A}$ succeeds with probability $1-f$ w.r.t.~$H$ for all subsets $H \subseteq N$ with $\sqrt{N} \leq |H| \leq N$.

Yao's minimax principle states that a lower bound of deterministic algorithms for any choice of input distribution implies a lower bound of randomized algorithms, so we have the following lemma.

\begin{lemma}\label{lem:rand-to-det}
For each multiset $\mathcal{H}$ of subsets $H \subseteq N$ with $\sqrt{N} \leq |H| \leq N$, we have \[T^\deterministic_{\mathcal{H}, f} \leq T^\randomized_{f}.\]
\end{lemma}
\begin{proof}
Let $\mathcal{A}^\randomized$ be a randomized algorithm that takes $T^\randomized_{f}$ time and succeeds with probability $1-f$ w.r.t.~$H$ for all subsets $H \subseteq N$ with $\sqrt{N} \leq |H| \leq N$. If we fix the random bits used by $\mathcal{A}^\randomized$, then we obtain a deterministic algorithm $\mathcal{A}^\deterministic$. Furthermore, the expected number of $H \in \mathcal{H}$ such that a collision-free transmission occurs when we run $\mathcal{A}^\deterministic$ on the subset of devices $H$ is at least $(1-f)|\mathcal{H}|$. Hence there exists a way of fixing the random bits such that $\mathcal{A}^\deterministic$ succeeds with probability $1-f$ w.r.t.~$\mathcal{H}$, so $T^\deterministic_{\mathcal{H}, f} \leq T^\randomized_{f}$.
\end{proof}

\paragraph{Hitting sets.}
We use the following  terminologies from \cite{ablp91,Newport14}.
 For any two sets $F \subseteq [N]$ and  $H \subseteq [N]$, we say that $F$ hits $H$ if $|F \cap H| = 1$. For any multiset $\mathcal{F}$ of subsets of $[N]$, we say that $\mathcal{F}$ hits the set $H$ if there exists a set $F \in \mathcal{F}$ that hits $H$.

 \begin{lemma}\label{lem:det-aux}
 $T^\deterministic_{\mathcal{H}, f} > t$ if there exists a multiset $\mathcal{H}$ of subsets $H \subseteq N$ with $\sqrt{N} \leq |H| \leq N$ such that each size-$t$ multiset $\mathcal{F}$ of subsets of $[N]$  hits at most a $(1-f)$ fraction of  $\mathcal{H}$.
 \end{lemma}
\begin{proof}
This lemma follows from the following interpretation: A $t$-time deterministic algorithm $\mathcal{A}$ can be interpreted as a size-$t$ multiset $\mathcal{F}=\{F_1, F_2, \ldots, F_t\}$ of subsets of $[N]$  by setting $F_i = \{ x \in [N] \ | \ a_i^x = \text{transmit}\}$. For any subset $H \subseteq N$ with $\sqrt{N} \leq |H| \leq N$, it is clear that  a collision-free transmission occurs when we run $\mathcal{A}$ on $H$ if and only if the corresponding multiset $\mathcal{F}$ hits $H$.
\end{proof}

The variant of \cref{lem:constant-f} without the constraint $\sqrt{N} \leq |H| \leq N$ was proved in~\cite{alon2014broadcast}, see~\cite{Newport14}. The same proof extends easily to our setting. For the sake of completeness, we still provide a proof  of \cref{lem:constant-f} here.

\begin{lemma}
\label{lem:constant-f}
Given an integer $N \geq 2$, there exists a multiset $\mathcal{H}$ of subsets $H \subseteq [N]$ with  $\sqrt{N} \leq |H| \leq N$  such that every subset $F \subseteq [N]$  hits at most  $O\left(\frac{1}{\log N}\right)$ fraction of $\mathcal{H}$.
\end{lemma}

\begin{proof}
The multiset  $\mathcal{H}$ is constructed by the following randomized procedure. For each integer $j \in [\log (\sqrt{N}), \log N]$, add $N^2$ random size-$2^j$ subset $H \subseteq [N]$ to $\mathcal{H}$, chosen uniformly at random from all size-$2^j$ subsets of $[N]$.

Consider any fixed $F \subseteq [N]$, write $p_n$ to denote the probability that $F$ hits a random size-$n$ subset $H \subseteq [N]$. 
\[p_n = 
\begin{cases}
\frac{|F| \cdot \binom{N-|F|}{n-1}}{\binom{N}{n}} & \text{if $|F| \leq N-n+1$,}\\
0 & \text{otherwise.}
\end{cases}
\]
According to the calculation in~\cite{alon2014broadcast}, we have 
\begin{align*}
 p_n &= O\left( \frac{n |F|}{N} \cdot \exp\left( -\frac{n |F|}{N} \right)\right),\\
 \sum_{j \in  [\log \sqrt{N}, \log N]} p_{2^j} &\leq  \sum_{j \in  [\log N]} p_{2^j} = O(1).
\end{align*}
Therefore, for any fixed $F \subseteq [N]$, the expected number of subsets in  $\mathcal{H}$ hit by $F$ is 
\[
N^2 \cdot \sum_{j \in  [\log \sqrt{N}, \log N]} p_{2^j} = O(N^2) = O\left(\frac{1}{\log N}\right) \cdot |\mathcal{H}|.
\]
By a Chernoff bound, this bound $O\left(\frac{1}{\log N}\right) \cdot |\mathcal{H}|$ holds with probability at least $1 - e^{-\Omega(N^2)}$. Therefore, by a union bound over all $2^N = e^{o(N^2)}$ choices of $F \subseteq [N]$, we conclude that every  subset $F \subseteq [N]$  hits at least at most  $O(1 / \log N)$ fraction of $\mathcal{H}$ with positive probability, so such a multiset  $\mathcal{H}$ exists.
\end{proof}

Now we are ready to prove \cref{lem:logn_time_lb}.

\begin{proof}[Proof of \cref{lem:logn_time_lb}]
Observe that \cref{lem:constant-f} implies that in order for a multiset $\mathcal{F}$ to hit at least a $(1-f)$ fraction of $\mathcal{H}$, then it is necessary that 
\[|\mathcal{F}| = \frac{1-f}{O(1 / \log N)} = \Omega((1-f) \log N).\]
For the case $f = 3/4$, we have $1-f = 1/4$, so 
\cref{lem:rand-to-det,lem:det-aux} imply that $T^\randomized_{f} \geq T^\deterministic_{\mathcal{H}, f} = \Omega(\log N)$.
To see that this bound implies \cref{lem:logn_time_lb}, we simply set $N = \tilde{n}$ and $f = 3/4$, as we observe that any lower bound in the setting of having identifiers also applies to the setting of having no identifiers.
\end{proof}

\end{document}